\documentclass{article}
\usepackage[utf8]{inputenc}

\usepackage[a4paper, margin=1in]{geometry}
\usepackage{setspace}
\usepackage{amsmath, accents, amsthm, bbm,enumitem, color, mathtools, nicefrac}
\usepackage[main=british]{babel}
\usepackage[babel=true]{csquotes}

\usepackage{lipsum}
\usepackage{amsfonts}
\usepackage{graphicx}
\usepackage{epstopdf}
\ifpdf
  \DeclareGraphicsExtensions{.eps,.pdf,.png,.jpg}
\else
  \DeclareGraphicsExtensions{.eps}
\fi

\usepackage{endnotes}
\let\footnote=\endnote

\usepackage[
backend = biber,
style = apa
]{biblatex}
\bibliography{references}

\newcommand{\E}{\mathbbm{E}}
\newcommand{\N}{\mathbbm{N}}
\newcommand{\R}{\mathbbm{R}}
\newcommand{\1}{\mathbbm{1}}

\newcommand{\Chi}{\mathcal{X}}

\newcommand{\la}{\lambda^{\ast}}
\newcommand{\Rex}{\bar{\R}}
\newcommand{\ubar}[1]{\underaccent{\bar}{#1}}
\newcommand{\diag}{\text{diag}}

\providecommand{\keywords}[1]
{
	\small	
	\textbf{\textit{Keywords---}} #1
}

\newtheorem{theorem}{Theorem}[section]
\newtheorem{definition}[theorem]{Definition}
\newtheorem{corollary}[theorem]{Corollary}
\newtheorem{lemma}[theorem]{Lemma}

\newtheorem{remark}[theorem]{Remark}

\newtheorem*{theorem*}{Theorem}
\newtheorem*{lemma*}{Lemma}
\newtheorem*{corrolary*}{Corollary}

\newtheorem{innercondition}{Condition}

\newenvironment{condition}[1]
  {\renewcommand\theinnercondition{#1}\innercondition}
  {\endinnercondition}

\allowdisplaybreaks

\title{Portfolio Optimization with Allocation Constraints and Stochastic Factor Market Dynamics}
\author{Marcos Escobar-Anel \and Michel Kschonnek \and Rudi Zagst}

\date{\today}

\begin{document}

\maketitle
\begin{abstract}
We study the expected utility portfolio optimization problem in an incomplete financial market where the risky asset dynamics depend on stochastic factors and the portfolio allocation is constrained to lie within a given convex set. We employ fundamental duality results from real constrained optimization to formally derive a dual representation of the associated HJB PDE. Using this representation, we provide a condition on the market dynamics and the allocation constraints, which ensures that the solution to the HJB PDE is exponentially affine and separable. This condition is used to derive an explicit expression for the optimal allocation-constrained portfolio up to a deterministic minimizer and the solution to a system of Riccati ODEs in a market with CIR volatility and in a market with multi-factor OU short rate. 
\end{abstract}
\keywords{Portfolio Optimization, Allocation Constraints, Stochastic Factor, Convex Duality, Exponential Affine Separability, Incomplete Markets, HJB}

\section{Introduction}\label{sec: Introduction}
In this paper, we consider a portfolio optimization problem of an investor who trades in continuous time and seeks to maximize his utility from terminal  wealth at the end of a finite investment horizon. The investor is assumed to be risk-averse and his risk-preferences are modelled by a power utility function. Our problem setting differs from the classic problem formulated in \cite{merton1969} with respect to two main aspects:\\

{\emph{(i) Market coefficients dependent on a stochastic factor.}}\\

Modelling stochastic market coefficients as a function of an additional stochastic factor is a natural extension to the classic Black-Scholes model which can capture some of the stylized facts observed in the financial market. One of the earliest discussions of such models in a portfolio optimization context was in \cite{Zariphopoulou2001}, where the author was able to characterize the solution to the associated HJB equation for a power-utiltiy function in terms of a linear parabolic PDE. Further, under the assumption of a global Lipschitz-condition on the market coefficients, a verification result was proven. However, explicit closed-form expressions for the opitmal allocation were only given when the stochastic factor is completely uncorrelated with the financial market, i.e. when the optimal allocation is myopic.  If the stochastic factor correlates with the financial market, closed-form expressions for the optimal allocation were recovered on individual occasions, e.g. in \cite{Deelstra2000} and \cite{Kraft2002} for financial markets with stochastic short rate and in \cite{Kraft2005} for financial markets with stochastic volatility. These advances required the solvability of certain underlying Riccati ODEs. The seminal work of \cite{Liu2006} unified these approaches by introducing a class of models where the asset returns have a quadratic dependence on the stochastic factor. Within such quadratic models, the author directly characterizes the HJB PDE as an exponentially quadratic function of the stochastic factor with coefficients determined by the solution to a system of Riccati ODEs. When reducing the framework of \cite{Liu2006} to an affine dependence on the stochastic factor, the results are closely related to affine term structure models of \cite{Duffie1996}. This affine reduction proved to be particularly fruitful for portfolio optimization applications, see e.g. \cite{Kallsen2010a}, \cite{baeuerle2013} and \cite{Rubtsov2017}. In addition, an extensive overview of related literature is given in \cite{Zariphopoulou2009}.\\
	 
{\emph{(ii) Convex constraints on relative portfolio allocation.}}\\
	
	In most practical applications investors need to follow allocation constraints either due to regulatory requirements or due to client preferences. \cite{Xu1990} and \cite{karatzas1991} first presented a duality approach for portfolio optimization problems with constraints on short-selling and trading of individual assets. Their duality approach was later generalized for general convex allocation constraints in \cite{cvitanic1992}, who derived a dual optimal control problem which seeks the least favorable market coefficients among a suitable set of \enquote{dual} stochastic processes. Employing martingale techniques, the equivalence of their dual problem to the considered primal constrained portfolio optimization problem was established for a general class of complete financial markets with suitably bounded stochastic market coefficients. Moreover, closed-form formulas for the optimal allocation were provided for deterministic market coefficients and a power utility function with risk aversion coefficient between $0$ and $1.$ Independently, \cite{Zariph1994} characterized the  value function for constraints on the wealth invested into risky assets via viscosity solutions to the associated HJB PDE, but did not provide explicit closed-form formulas for the associated optimal allocation. Allocation constraints have been considered in a variety of different contexts since (see e.g. \cite{Cuoco1997}, \cite{Bouchard2004}, \cite{Zheng2011}, \cite{Nutz2012}, \cite{Larsen2013}, \cite{Dong2018}, \cite{Dong2020}), but explicit solutions and theoretical guarantees for the corresponding optimal allocation (especially in the non-myopic case) have remained scarce and solutions need to be estimated by suitable numerical schemes for many advanced models (e.g. \cite{Zheng2021}, \cite{Zhu2022}). \\
	
Aspects (i) and (ii) result in the investor not being able to replicate all measurable payoffs at the end of the investment horizon and therefore standard martingale techniques cannot be employed to characterize the value function of the optimization problem. Further, it is unclear if  the Hamilton-Jacobi-Bellman PDE (\enquote{HJB} PDE) of the optimization problem admits a smooth solution due to pointwise constraints on the optimal relative portfolio allocation. Both aspects have only been studied simultaneously on rare occasions. \cite{Pham2002} used a logarithmic transformation to characterize the solution to the constrained HJB PDE through the solution to a semilinear PDE. Assuming Lipschitz- and non-degeneracy conditions on the stochastic factor, the existence of a smooth solution to the transformed PDE can be guaranteed and closed-form expressions for the optimal portfolio can be given if the stochastic factor is uncorrelated with the financial market. In \cite{Detemple2005}, semi-closed-form expressions are provided for the optimal allocation for a market with generalized Vasicek short rate and bounds on the portfolio allocation to a hedging instrument for interest-rate risk. \cite{Mnif2007} and \cite{Mnif2011} develop numerical schemes for allocation constrained portfolio optimization problems in jump-diffusion models, where the asset volatilities and jumps depend on an external stochastic factor.\\

Our contribution to this literature is threefold. First, we present an approach to constrained portfolio optimization that transforms the HJB PDE associated with the constrained portfolio optimization problem into an equivalent dual PDE, which is the HJBI PDE associated with a dual minimzation problem akin to Condition (C) in \cite{cvitanic1992}. However, unlike in \cite{cvitanic1992}, the validity of this method is not tied to the completeness of the underlying financial market and can thus be applied in a broader context. Secondly, in the spirit of \cite{Liu2006}, we derive a condition on the dynamics of the financial market and the allocation constraints, which ensures that the value function of the optimal investment problem is exponentially affine. Lastly, we provide expressions for the allocation constrained optimal allocation in a market with multi-factor stochastic volatility of CIR-Type and multi-factor short rate of OU-type. These expressions are explicit up to a deterministic minimizer and the solution of a system of Riccati ODEs, which leads to a non-myopic optimal allocation if the stochastic factor correlates with the financial market. In particular, the optimal allocation is generally non-myopic. \\

The remainder of this paper is structured as follows: The financial market model and the portfolio optimization problem $\mathbf{(P)}$ are introduced in Section \ref{sec: Prerequisites}. Afterwards, in Section \ref{sec: The Dual HJBI PDE}, we use a result from real constrained optimization to show that the constrained HJB PDE associated with $\mathbf{(P)}$ is equivalent to the Hamilton-Jacobi-Bellman-Isaacs PDE (\enquote{HJBI PDE}) associated with a dual minimax problem and derive a condition under which the solution to both PDEs is exponentially affine. The versatility and use of the derived condition is illustrated in examples with deterministic market coefficients, stochastic volatility and stochastic short rate in Section \ref{sec: Examples}. Finally, Section \ref{sec: conclusion} concludes the paper. All proofs of Theorems, Lemmas and Corollaries in this paper can be found in the appendix.

\section{Prerequisites}\label{sec: Prerequisites}

We consider a finite time horizon $T>0$ and a complete, filtered probability space $(\Omega, \mathcal{F}_T,\mathbbm{F} = (\mathcal{F}_t)_{t \in [0,T]}, Q)$, where the filtration $\mathbbm{F}$ is generated by the independent $m$-dimensional Wiener process $W^z = \left(W^z(t)\right)_{t \in [0,T]}$ and  $d$-dimensional Wiener process $\hat{W} = \big(\hat{W}(t)\big)_{t\in [0,T]}$. On this probability space, we model a financial market consisting of $d$ assets, whose dynamics are dependent on the value of a another stochastic process $z$. To this end, we consider a constant $z_0 \in \R$ and the deterministic functions $\mu^z: [0,T]\times \R^m \rightarrow \R^m$ and $\Sigma^z: [0,T]\times \R^m \rightarrow \R^{m \times m}$, which are assumed to be sufficiently regular such that the SDE 
$$
dz(t) = \mu^z(t,z(t))dt + \Sigma^z(t,z(t))' dW^z(t), \quad z(0) = z_0, \ t\in[0,T]
$$
admits a solution $z = \left(z(t)\right)_{t\in [0,T]}$. In addition, we consider another deterministic function $\rho:[0,T]\times\R^m\rightarrow \R^{m\times d}$ with columns $\rho_i$, $i=1,...,d$ satisfying $\Vert \rho_i(t,x)\Vert \leq 1$ for all $x\in \R^m$.\footnote{We use $\Vert \cdot \Vert$ to denote the standard Euclidian norm.} The function $\rho$ is used to control the correlation between the diffusion driving $z$ and the diffusion driving the financial market. To this end, we define the stochastic processes $W_i = \left(W_i(t)\right)_{t\in [0,T]}$ as
$$
W_i(t) = \rho_i(t,z(t))'W^z(t) + \sqrt{1-\Vert \rho_i(t,z(t))\Vert^2}\hat{W}_i(t), \quad i=1,...,d.
$$ 
Then, by construction $W = (W_1,...,W_d)'$ is a $d$-dimensional Wiener process such that
$$
d\langle W^z_i, W_j\rangle_t = \rho_{ij}(t,z(t))dt.
$$
Finally, we consider three additional deterministic functions $r: [0,T]\times \R^m \rightarrow \R$, $\mu:[0,T]\times\R^m \rightarrow \R^d$ and $\Sigma:[0,T]\times\R^m \rightarrow \R^{d\times d}$ and define our financial market model $\mathcal{M}$ consisting of one risk-free asset $P_0$ and $d$ risky assets $P = (P_1,...,P_d)'$ through the dynamics
$$
dP_0(t) = P_0(t) \cdot r(t,z(t)) dt, \qquad P_0(0) = 1 
$$
and
$$
dP(t)  = \diag(P(t)) \cdot \left[\mu(t,z(t))dt + \Sigma(t,z(t))dW(t)\right], \quad P(0)= \1 \in \R^d.
$$
Again, we assume that $r,\ \mu,$ and $\Sigma$ are sufficiently regular so that solutions to the above SDEs exist. Moreover, we assume that $\Sigma$ is chosen such that
$$
\Sigma(t,z(t)) \ \text{is} \ \mathcal{L}[0,T]\otimes Q-\text{a.e.} \ \text{non-singular.}
$$
If the context is unambiguous, we only write $\mu^z,$ $\Sigma^z,$ $\rho,$ $r,$ $\mu,$ $\Sigma$ instead of $\mu^z(t,z(t)),$ $\Sigma^z(t,z(t)),$ $\rho(t,z(t)),$ $r(t,z(t)),$ $\mu(t,z(t)),$ $\Sigma(t,z(t))$ to improve the clarity of presentation.
 
The wealth process $V^{v_0, \pi}$ of an investor with initial wealth $v_0>0$ and trading in $\mathcal{M}$ according to a $d$-dimensional relative portfolio process $\pi$ satisfies the usual SDE
\begin{align}\label{eq: SDE V original market}
V^{v_0, \pi}(0) &= v_0 \nonumber \\
dV^{v_0, \pi}(t) &= V^{v_0, \pi}(t) \left(\left[r + (\mu-r\1)'\pi(t) \right]dt + \pi(t)'\Sigma dW(t) \right)
\end{align}

In this context, the relative portfolio allocation process $\pi= \big(\pi(t)\big)_{t\in [0,T]}$ is a $d$-dimensional process, where $\pi_i(t)$ denotes the fraction of wealth invested in the risky asset $P_i$ at time $t$. The remaining fraction $1-\1'\pi(t)$ is invested in the risk-free asset $P_0$. We restrict our analysis to the portfolio processes $\pi$, which guarantee that a unique, strictly positive solution to (\ref{eq: SDE V original market}) exists i.e. to $\pi$ in 
\begin{align}\label{eq: def. admissible strategy Heston}
    \Lambda = \Big \{ \pi = \big(\pi(t)\big)_{t \in [0,T]} \ \text{progr. measurable} \  \Big | \  \int_0^T \Vert \Sigma'\pi(t) \Vert^2 dt < \infty \ Q-a.s.\Big \}
\end{align}
If $\pi \in \Lambda$, it is straightforward to show that the unique solution $V^{v_0, \pi}(t)$ to (\ref{eq: SDE V original market}) can be expressed in closed-form as 
\begin{align}\label{eq: wealth process explicitly original market}
V^{v_0, \pi}(t) = v_0 &\exp \Big( \int_0^t r+ (\mu - r\1)' \pi(s) - \frac{1}{2}\Vert \Sigma'\pi(s)\Vert ^2 ds  + \int_0^t \pi(s)'\Sigma dW(s)\Big). \nonumber
\end{align}

Below, we often work with so-called Markovian controls $\pi\in \Lambda$, which are defined in feedback-form
$$\pi(t) = \ubar{\pi}(t,V^{v_0,\pi}(t),z(t)),$$
for a deterministic measurable function $\ubar{\pi}:[0,T]\times (0,\infty) \times \R^m \rightarrow \R^d$ of the current state $(t,V^{v_0,\pi}(t),z(t))$ of the financial market. As these processes are uniquely defined through the function $\ubar{\pi}$, we use the process $\pi$ and the function $\ubar{\pi}$ interchangeably. When a process is specifically defined in feedback-form, we follow the notation in \cite{Soner2006} and denote this fact by a \enquote{lower bar}, i.e. $\ubar{\pi}$. The set of admissible Markovian controls is thus denoted by $\ubar{\Lambda}$ and the wealth-process corresponding to a Markovian control $\ubar{\pi}$ is denoted by $V^{v_0, \ubar{\pi}}$. The remaining notation will carry over analogously. \\

For a closed convex set $K\subset \left(\R\cup\{\infty, -\infty\}\right)^d=:\Rex^d$ with non-empty interior and CRRA-utility function $U(v)= \frac{1}{b}v^b$ with $b<1$ and $b\neq 0$, our investor faces an allocation constrained primal portfolio optimization $\mathbf{(P)}$ of the form
\begin{equation*}
\mathbf{(P)}
\begin{cases}
 \Phi(v_0) &= \underset{\pi \in \Lambda_K}{\sup} \mathbbm{E}\big[U(V^{v_0, \pi}(T)) \big] \\
 \Lambda_K &= \big \{ \pi(t) \in K \ \mathcal{L}[0,T]\otimes Q-\text{a.e.} \ \big | \ \pi \in \Lambda \big \}
\end{cases}
\end{equation*}

We approach $(\mathbf{P})$ using classic methods from stochastic optimal control. For this purpose, let us introduce the  generalized primal portfolio optimization problem $\mathbf{(P^{(t,v,z)})}$ as 
\begin{equation*}
	\mathbf{(P^{(t,v,z)})}
	\begin{cases}
		\Phi(t,v,z) &= \underset{\pi \in \Lambda_K(t)}{\sup} \mathbbm{E}\big[U(V^{v_0, \pi}(T))\  \big | \ V^{v_0, \pi}(t) = v, \ z(t) = z \big] \\
		\Lambda_K(t) &= \big \{ \big(\pi(s)\big)_{s\in [t,T]} \ \big | \ \pi \in \Lambda_K \big \}.
	\end{cases}
\end{equation*} 
Then, the Hamilton-Jacobi-Bellman equation (\enquote{HJB equation}) associated with $\mathbf{(P^{(t,v,z)})}$ is given by
\begin{alignat}{2}\label{eq: constr. HJB PDE Factor Model}
	0&= \sup_{\pi \in K}\Big \{  G_t + v\left[r+ (\mu - r \1)'\pi\right]G_v +  \frac{1}{2}v^2 \Vert \Sigma'\pi \Vert^2 G_{vv} + \left(\mu^z\right)'\left(\nabla_z G\right) \nonumber \\
	& \qquad \qquad + v\left(\Sigma^z \rho \Sigma'\pi\right)' \nabla_z \left(G_v\right)+\frac{1}{2}\text{Trace}\left[\Sigma^z\left( \Sigma^z\right) '\nabla^2_zG\right]\Big \} \nonumber \\
	&=  G_t + v r G_v + \left(\mu^z\right)'\left(\nabla_z G\right) + \frac{1}{2}\text{Trace}\left[\Sigma^z\left( \Sigma^z\right) '\nabla^2_zG\right] \nonumber \\*
	& \quad + v \sup_{\pi \in K}\Big \{  (\mu - r \1)'\pi G_v +  \left(\Sigma^z \rho \Sigma'\pi\right)' \nabla_z \left(G_v\right) +  \frac{1}{2}v \Vert \Sigma'\pi \Vert^2 G_{vv}       \Big \} \\
	G(T,v,z) &= U(v), \label{eq: terminal condition constr. HJB equation}
\end{alignat}
Any (sufficiently regular) solution $G$ to (\ref{eq: constr. HJB PDE Factor Model}) yields a candidate optimal Markovian control through the maximizing argument 
\begin{align}\label{eq: candidate optimal allocation}
\ubar{\pi}^{\ast}(t,v,z) = \underset{\pi \in K}{\text{argmax}}\Big \{  (\mu - r \1)'\pi G_v +  \left(\Sigma^z \rho \Sigma'\pi\right)' \nabla_z \left(G_v\right) +  \frac{1}{2}v \Vert \Sigma'\pi \Vert^2 G_{vv}\Big \}.
\end{align}
Even without the additional presence of allocation constraints, it is notoriously difficult to characterize the solution $G$ to the HJB PDE (\ref{eq: constr. HJB PDE Factor Model}) and even more challenging to determine an explicit expressions for $G$. For this reason, we devote the upcoming Section \ref{sec: The Dual HJBI PDE} to deriving an equivalent dual representation of (\ref{eq: constr. HJB PDE Factor Model}), which leads to a dual approach to solving the allocation constrained portfolio optimization problem $(\mathbf{P})$.

\section{The Dual HJBI PDE}\label{sec: The Dual HJBI PDE}


\subsection{Dual Approach to Constrained Optimization over $\R^d$}\label{subsec: Dual Approach to Constrained Optimization}
In this subsection, we summarize a selection of the duality results from \cite{Rockafellar1974} and prove a duality statement for general constrained quadratic maximization problems over $\R^d$ in Lemma \ref{lem: constrained real optimization}, which forms the basis of our dual approach to solving $\mathbf{(P)}$. We would like to emphasize that the results presented in this subsection constitute only a tiny fraction of the more general methodology presented in \cite{Rockafellar1974} and generalizations can be made easily.\\

Consider a proper\footnote{A proper concave function $f: \R^d \rightarrow \Rex$ is a concave function with $f(x)<\infty$ for all $x\in \R^d$ and there exists at least one $x_0 \in \R^d$ such that $f(x_0)> -\infty$.} real concave function $f:\R^d \rightarrow \Rex$ and a function ${F:\R^d \times \R^d\rightarrow \Rex}$ with $F(x,0) = f(x)$ for all $x\in \R^d$. The function $F(x,u)$ can be regarded as a pertubed version of $f(x)$, with pertubation parameter $u$.\footnote{It is important to emphasize that the ensuing duality relations as well as their applicability are dependent on the specific choice of the pertubation $F(x,u)$ of $f(x)$.} Moreover, we define the value function of a maximization problem over $F$, resp. $f$ as
\begin{align*}
    \Phi(u):= \sup_{x\in \R^d}F(x,u) \qquad \Phi_P:= \Phi(0) = \sup_{x\in \R^d}\underbrace{F(x,0)}_{=f(x)}=   \sup_{x\in \R^d}f(x).
\end{align*}
Moreover, let us introduce the (concave) conjugate and bi-conjugate of $F$ w.r.t. $u$ as 
\begin{itemize}
\item $F^{\ast}(x,\lambda)= \sup_{u \in \R^d}\big( F(x,u) -  \lambda' u  \big)$ for $(x,\lambda) \in \R^d\times \R^d$
\item $F^{\ast \ast}(x,u) = \inf_{\lambda \in \R^d}\big(  F^{\ast}(x,\lambda) + \lambda' u \big)$ for $(x,u)\in \R^d\times \R^d$.
\end{itemize}
 The concave conjugate and bi-conjugate naturally appear, when characterizing the u.s.c. concave hull of $F(x,\cdot)$ as the closure of the concave hull of the epigraph of $F(x,\cdot)$, which in turn can be obtained as the intersection of half spaces that contain it. It can be shown that $F^{\ast \ast}(x,\cdot)$ is the u.s.c. concave hull of $F(x,\cdot):\R^d\rightarrow \Rex$ (Theorem 5 in \cite{Rockafellar1974}). In particular, we have $F^{\ast \ast}(x,u) = F(x,u)$ for all $u\in \R^d$ if and only if $F(x,\cdot)$ is u.s.c. concave in $u$.\footnote{
In the convex analysis literature (see e.g. equation (3.25) in \cite{Rockafellar1974}), it is customary to define the conjugate of $F$ in the concave sense (w.r.t. $u$) as
$$F^{\ast}(x,\lambda) = - \sup_{u \in \R^d}\big( F(x,u) -  \lambda' u\big) = \inf_{u \in \R^d}\big( \lambda' u  - F(x,u) \big).$$
In contrast to our definition, this has the satisfying consequence that the bi-conjugate of $F(x,\cdot)$ is obtained by taking the conjugate of $F^{\ast}(x,\cdot)$, i.e. $F^{\ast \ast}(x,\cdot) = (F^{\ast}(x,\cdot))^{\ast}(x,\cdot)$. 
According to our definition, we have the (slightly) less elegant version $F^{\ast \ast}(x,\cdot) =-(-F^{\ast}(x,\cdot))^{\ast}(x,\cdot)$. However, in the mathematical finance literature, in particular in a portfolio optimization context, our definition is more prevalent. Since we are going to apply the above theory in a portfolio optimization context, we decided to adhere to the corresponding convention.
} Investigation of the latter relation between $F$ and its bi-conjugate $F^{\ast \ast}$ explicitly for $u=0$ leads to the duality framework presented in \cite{Rockafellar1974}. For this purpose, we define the Lagrangian $L$ as $L(x,\lambda) = F^{\ast}(x,\lambda)$ and the optimal value of the dual optimization problem as
$$\Psi_D := \inf_{\lambda \in \R^d} \sup_{x\in \R^d}\big( L(x,\lambda) \big).$$
Moreover, if $F(x,\cdot):\R^d\rightarrow \Rex$ is proper, u.s.c. and concave $\forall x\in \R^d$, then we may equivalently write
$$\Phi_P = \sup_{x\in \R^d}F(x,0) = \sup_{x\in \R^d}F^{\ast \ast}(x,0) = \sup_{x\in \R^d} \inf_{ \lambda \in \R^d}L(x,\lambda).$$
In particular, so-called weak duality always holds in this context
\begin{align*}
\Phi_P = \sup_{x\in \R^d} \inf_{ \lambda \in \R^d}L(x,\lambda) \leq \inf_{\lambda  \in \R^d} \sup_{x \in \R^d} L(x,\lambda) = \Psi_D,    
\end{align*}
which can even be strengthened to strong duality $(\Phi_P = \Psi_D)$ if we can find a so-called saddle-point $(x^{\ast},\lambda^{\ast})$, which satisfies
$$L(x,\lambda^{\ast}) \leq L(x^{\ast},\lambda^{\ast}) \leq L(x^{\ast},\lambda)\quad \forall (x,\lambda)\in \R^d\times \R^d.$$

\begin{lemma}\label{lem: saddle point implies strong duality}
Let $F(x^{\ast},\cdot)$ be u.s.c. and concave in $u$ for a specific $x^{\ast}\in \R^d$. Then, the following two statements are equivalent:
\begin{itemize}
    \item[(i)] $x^{\ast}$ is optimal for $\mathbf{(P)}$, $\lambda^{\ast}$ is optimal for $\mathbf{(D)}$ and $\Phi_P = \Psi_D$
    \item[(ii)] $(x^{\ast},\lambda^{\ast})$ is a saddle-point of the Lagrangian $L$.
\end{itemize}
\end{lemma}

Note however that the concept of a saddle-point only provides a sufficient, but not a necessary condition for strong duality to hold. In particular, there are several other sufficient conditions for strong duality - one of which is the Slater condition.
\begin{lemma}[Slater's condition]\label{lem: Slater's condition}
\textcolor{white}{1}\\
Assume $F$ is concave jointly in $(x,u)$ and there exists an $\hat{x}\in \R^d$ such that $F(\hat{x}, u)$ is bounded below on a neighborhood of $u=0$. Then, strong duality holds, i.e.
$$\Phi_P = \Psi_D.$$
\end{lemma}

By using Lemma \ref{lem: saddle point implies strong duality} and Lemma \ref{lem: Slater's condition}, we can prove the following statement about constrained quadratic optimization over $\R^d$. This result forms the basis for the duality approach presented in Subsection \ref{subsec: Dual Approach to Allocation-Constrained Portfolio Optimization}.

\begin{lemma}\label{lem: constrained real optimization}
	Let $K\subset \Rex^d$ be closed convex with non-empty interior and let $\delta_K(x)= -\inf_{y\in K}(x'y)$ be the support function of $K.$ Further, consider a proper u.s.c. concave function $\tilde{f}:\R^d\rightarrow \Rex$ which admits a unique maximizer $x^{\ast}$ of $\tilde{f}$ over $K.$
	Then,
	\begin{align*}
		\sup_{x\in K}\tilde{f}(x) = \inf_{\lambda \in \R^d}\sup_{x\in \R^d} \Big  \{ \tilde{f}(x) + \delta_K(\lambda) + x'\lambda \Big \}.
	\end{align*}
	Further, if $\la$ minimizes $\sup_{x\in \R^d} \big \{ \tilde{f}(x) + \delta_K(\lambda) + x'\lambda \big \}$ and $\tilde{f}(x)+x'\la$ admits a unique maximizer $x \in \R^d$, then
	\begin{align}
		&\sup_{x\in \R^d} \big \{\tilde{f}(x) + \delta_K(\lambda^{\ast}) + x'\lambda^{\ast} \big \} = \tilde{f}(x^{\ast}) + \delta_K(\lambda^{\ast}) + \left(x^{\ast}\right)'\lambda^{\ast} \label{eq: optimal constrained x as unconstrained maximizer of general dual 1} \\
		\Leftrightarrow \qquad & \tilde{f}(x^{\ast}) = \sup_{x\in K}\tilde{f}(x) \quad \text{and} \quad x^{\ast} \in K. \label{eq: optimal constrained x as unconstrained maximizer of general dual 2}
	\end{align}
\end{lemma}

\begin{remark}\label{rem: constrained quadratic optimization}
In particular, any quadratic function $\tilde{f}(x) = -x'Ax + b'x + c,$ with symmetric, positive definite matrix $A\in \R^{d\times d},$ $b\in \R^d$ and $c\in \R$ satisfies the requirements of Lemma \ref{lem: constrained real optimization}, due to the strict concavity of $\tilde{f}(x)+x'\lambda$ and $\tilde{f}(x)+x'\lambda\rightarrow -\infty$ for $\Vert x \Vert \rightarrow \infty$ and any $\lambda \in \R^d.$
\end{remark}

\begin{remark}\label{rem: slackness condition}
Combining equations (\ref{eq: optimal constrained x as unconstrained maximizer of general dual 1}) and (\ref{eq: optimal constrained x as unconstrained maximizer of general dual 2}) yields the complementary slackness condition 
\begin{align*}
	\tilde{f}(x^{\ast}) &= \sup_{x\in K}\tilde{f}(x) = \inf_{\lambda \in \R^d}\sup_{x\in \R^d} \Big  \{ \tilde{f}(x) + \delta_K(\lambda) + x'\lambda \Big \} \\
	&= \sup_{x\in \R^d} \big \{\tilde{f}(x) + \delta_K(\lambda^{\ast}) + x'\lambda^{\ast} \big \} = \tilde{f}(x^{\ast}) + \delta_K(\lambda^{\ast}) + \left(x^{\ast}\right)'\lambda^{\ast} \\
	\Leftrightarrow 0 &= \delta_K(\lambda^{\ast}) + \left(x^{\ast}\right)'\lambda^{\ast}.
\end{align*}
\end{remark}


\subsection{Dual Approach to Allocation Constrained Portfolio Optimization}\label{subsec: Dual Approach to Allocation-Constrained Portfolio Optimization}
In this subsection, we use Lemma \ref{lem: constrained real optimization} to derive an equivalent dual representation of the HJB equation (\ref{eq: constr. HJB PDE Factor Model}), which can be regarded as the Hamilton-Jacobi-Bellman-Isaacs equation (\enquote{HJBI equation}) associated with a dual optimization problem over a certain class of stochastic processes. This dual optimization problem as well as our solution approach to constrained portfolio optimization will closely resemble optimality Condition (C) from \cite{cvitanic1992}. However, unlike \cite{cvitanic1992}, we arrive at this approach applying duality arguments directly to the pointwise optimization on the level of the HJB PDE (\ref{eq: constr. HJB PDE Factor Model}) rather than on the level of stochastic processes. This reduces the level of technicality involved and removes the necessity for market completeness as a central underlying assumption.

\begin{lemma}\label{lem: dual HJB PDE}{(Dual HJBI PDE)} 
\textcolor{white}{1}\\
Let  $G \in C^{(1,2,2)}([0,T]\times (0,\infty)\times \R^m)$ be strictly concave and strictly increasing in the second component $v.$ Then, $G$ is a solution to (\ref{eq: constr. HJB PDE Factor Model}) if and only if $G(T,v,z)=U(v)$ and 
\begin{align}\label{eq: dual HJB PDE general}
    0= G_t &+ v r G_v + \left(\mu^z\right)'\left(\nabla_z G\right) + \frac{1}{2}\text{Trace}\left[\Sigma^z\left( \Sigma^z\right) '\nabla^2_zG\right] \nonumber \\
    & + v \inf_{\lambda \in \R^d}\sup_{\pi \in \R^d}\Big \{ \left[ \delta_K(\lambda) +(\mu + \lambda - r \1)'\pi \right]G_v   +  \left(\Sigma^z \rho \Sigma'\pi\right)' \nabla_z \left(G_v\right) +  \frac{1}{2}v \Vert \Sigma'\pi \Vert^2 G_{vv} \Big \}  \\
   = G_t &+ v r G_v + \left(\mu^z\right)'\left(\nabla_z G\right) + \frac{1}{2}\text{Trace}\left[\Sigma^z\left( \Sigma^z\right) '\nabla^2_zG\right] \nonumber \\
   	& + v \inf_{\lambda \in \R^d}\Big \{ \delta_K (\lambda) G_v- \frac{1}{2}\frac{1}{vG_{vv}}\Vert \Sigma^{-1}\left[\mu + \lambda - r\1\right]G_v + \left(\Sigma^z\rho\right)'\nabla_z\left(G_v\right)\Vert^2 \Big \}. \nonumber
\end{align}
\end{lemma}

The dual PDE (\ref{eq: dual HJB PDE general}) is the  HJBI PDE {(see e.g. Section 4.2 in \cite{Isaacs1999}) associated with the minimax stochastic control problem 
\begin{align}\label{eq: dual optimization problem (general)}
    \inf_{\lambda \in \mathcal{D}}\sup_{\pi \in \Lambda} \mathbbm{E}\Big[ U\Big(\underbrace{V^{v_0,\pi}(T)\cdot \exp\Big(\int_0^T \lambda(t)'\pi(t)+\delta_K(\lambda(t))dt\Big)}_{=: V^{v_0,\pi}_{\lambda}(T)}\Big)\Big],
\end{align}
where the dual control $\lambda = (\lambda(t))_{0\leq t \leq T}$ is taken from a suitable space of progressively measurable processes
\begin{align*}
    \mathcal{D} =  \Big \{ &\lambda = \big(\lambda(t)\big)_{t \in [0,T]} \ \text{progr. measurable} \  \Big |  \  \int_0^T \Vert\lambda(t)\Vert^2+\delta_K(\lambda(t))  dt < \infty \ Q-a.s.\Big \}
\end{align*}
Just as with portfolio processes, we refer to dual processes defined in feedback form $\lambda(t) = \ubar{\lambda}(t,V^{v_0,\pi}_{\lambda}(t),z(t))$, for a deterministic measurable function $\ubar{\lambda}$, as Markovian dual controls. Analogously, when specifically referring to Markovian dual controls, we write \enquote{$\ubar{\lambda}$} and collect all admissible Markovian controls in $\ubar{\mathcal{D}}.$ \\

Although intuitively appealing, the relationship between HJB(I) PDEs and the associated optimization problems still requires formal mathematical justification via verification theorems. Due to the generality of the setting considered in this work, we can thus only provide general verification theorems under additional assumptions on the candidate optimal controls $\ubar{\pi}^{\ast}$ (and $\ubar{\lambda}^{\ast}),$ the solution $G$ to the HJB(I) PDE and the financial market $\mathcal{M}$. Verifying such conditions is typically only feasible in more narrowly focussed settings (see e.g. Corollary \ref{cor: uniform integrability B.S. model} and \ref{cor: uniform integrability OU model} in Section \ref{sec: Examples}).  \\
Here, we make the relation between the PDE (\ref{eq: dual HJB PDE general}) and the dual control problem (\ref{eq: dual optimization problem (general)}) more precise, by proving a verification theorem, which relies on an additional uniform integrability condition \ref{cond: uniform integrability condition for dual} (compare to e.g. Definition 4.2 in \cite{Kraft2005}). 

\begin{condition}{($\text{UI}_{\lambda}$)}\label{cond: uniform integrability condition for dual}
For given $n\in \N$, $t\in [0,T]$,  $G \in C^{(1,2,2)}([0,T]\times (0,\infty)\times \R^m)$, $\lambda \in \mathcal{D}$, $\pi \in \Lambda$ we define the stopping time $\tau^{\lambda}_{n,t}=\min(T,\hat{\tau}^{\lambda}_{n,t})$, with
\begin{align*}
	\hat{\tau}^{\lambda}_{n,t} = \inf \Big \{t\leq u \leq T \ \Big | \ &\int_t^u \left(V_{\lambda}^{v_0, \pi}(s)\cdot \Vert \Sigma(s,z(s))'\pi(s)\Vert\cdot G_v(s,V_{\lambda}^{v_0, \pi}(s),z(s))\right)^2ds \geq n, \\ 
	&\int_t^u \Vert \left(\Sigma^z(s,z(s)\right)'\nabla_z\left(G\right)(s,V_{\lambda}^{v_0, \pi}(s),z(s))\Vert^2ds \geq n \Big \}.
\end{align*}
We say that $G,$ $\pi,$ and $\lambda$ satisfy condition \ref{cond: uniform integrability condition for dual} if for every $t\in [0,T]$, the sequence $\big(G(\tau^{\lambda}_{n,t},V_{\lambda}^{v_0, \pi}(\tau^{\lambda}_{n,t}),z(\tau^{\lambda}_{n,t}))\big)_{n\in \N}$ is uniformly integrable.
\end{condition}

\begin{remark}\label{rem: consequance condition UI}
	According to Theorem 4.5.4 in \cite{Bogachev2007}, if $G,$ $\pi,$ and $\lambda$ satisfy Condition \ref{cond: uniform integrability condition for dual}, then we have for every $t\in [0,T]$ that $\tau_{n,t}\rightarrow T$ $Q$-a.s., as $n\rightarrow \infty$ and 
	\begin{align}\label{eq: dominated convergence from condition UI}
		\mathbbm{E} \Big[G(T,V_{\lambda}^{v_0, \pi}(T),z(T)) \ \Big |  \mathcal{F}_t \Big] &= \mathbbm{E} \Big[  \lim_{n\rightarrow \infty}G(\tau_{n,t},V_{\lambda}^{v_0, \pi}(\tau_{n,t}),z(\tau_{n,t})) \ \Big |  \mathcal{F}_t \Big] \nonumber \\
		&=\lim_{n\rightarrow \infty}\mathbbm{E} \Big[G(\tau_{n,t},V_{\lambda}^{v_0, \pi}(\tau_{n,t}),z(\tau_{n,t})) \ \Big | \mathcal{F}_t \Big].
	\end{align}
\end{remark}

\begin{lemma}[Verification Theorem Dual Control Problem]\label{lem: weak verification theorem dual HJB PDE}
\textcolor{white}{1}\\
Let  $G \in C^{(1,2,2)}([0,T]\times (0,\infty)\times \R^m)$ be a solution to the dual HJB equation (\ref{eq: dual HJB PDE general}), be non-negative,  strictly concave and increasing in $v$. Let the feedback controls $\ubar{\lambda}^{\ast}(t,v,z)$, $\ubar{\pi}^{\ast}(t,v,z)$ be such that for all $(t,v,z)\in [0,T]\times (0,\infty)\times \R^m$
\begin{align*}
	&\inf_{\lambda \in \R^d}\sup_{\pi \in \R^d}\Big \{ \left[ \delta_K(\lambda) +(\mu + \lambda - r \1)'\pi \right]G_v  +  \left(\Sigma^z \rho \Sigma'\pi\right)' \nabla_z \left(G_v\right) +  \frac{1}{2}v \Vert \Sigma'\pi \Vert^2 G_{vv} \Big \} \\
	= &  \left[\delta_K(\ubar{\lambda}^{\ast}(t,v,z)) +(\mu + \ubar{\lambda}^{\ast}(t,v,z) - r \1)'\ubar{\pi}^{\ast}(t,v,z)\right]G_v  \\
	&  \qquad \qquad \qquad \quad  +  \left(\Sigma^z \rho \Sigma'\ubar{\pi}^{\ast}(t,v,z)\right)' \nabla_z \left(G_v\right) +  \frac{1}{2}v \Vert \Sigma'\ubar{\pi}^{\ast}(t,v,z) \Vert^2 G_{vv}.
\end{align*}
Then the following holds $\forall (t,v,z)\in [0,T]\times (0,\infty)\times \R^m$:
\begin{itemize}
    \item[(i)] If $(\ubar{\lambda}^{\ast},\pi)\in \ubar{\mathcal{D}}\times \Lambda$ satisfy condition \ref{cond: uniform integrability condition for dual}, then 
\begin{align}\label{eq: dual optimality la}
  G(t,v,z) \geq  \mathbbm{E}\big[U(V_{\ubar{\lambda}^{\ast}}^{v_0, \pi}(T))\  \big | \ V_{\ubar{\lambda}^{\ast}}^{v_0, \pi}(t) = v, \ z(t) =z \big].  
\end{align}
    \item[(ii)] If $(\lambda,\ubar{\pi}^{\ast})\in \mathcal{D}\times \ubar{\Lambda}$ satisfy condition \ref{cond: uniform integrability condition for dual}, then 
\begin{align}\label{eq: dual optimality pi ast}
G(t,v,z) \leq  \mathbbm{E}\big[U(V_{\lambda}^{v_0,\ubar{\pi}^{\ast}}(T))\  \big | \ V_{\lambda}^{v_0,\ubar{\pi}^{\ast}}(t) = v, \ z(t) =z \big].
\end{align}
    \item[(iii)] If $(\ubar{\lambda}^{\ast},\ubar{\pi}^{\ast})\in \ubar{\mathcal{D}}\times \ubar{\Lambda}$ satisfy condition \ref{cond: uniform integrability condition for dual}, then 
\begin{align}\label{eq: dual equality la and pi}
G(t,v,z) =  \mathbbm{E}\big[U(V_{\ubar{\lambda}^{\ast}}^{v_0,\ubar{\pi}^{\ast}}(T))\  \big | \ V_{\ubar{\lambda}^{\ast}}^{v_0,\ubar{\pi}^{\ast}}(t) = v, \ z(t) =z \big].
\end{align}
\end{itemize}
\end{lemma}

\begin{remark}
If we restrict the minimization and maximization in the min-max optimization only to such $\lambda \in \mathcal{D}_{\text{UI}} \subset \mathcal{D}$, $\pi \in \Lambda_{\text{UI}}\subset \Lambda$ so that every pair $(\lambda, \pi)$, $(\ubar{\lambda}^{\ast},\pi)$, $(\la,\ubar{\pi}^{\ast})$ and $(\ubar{\lambda}^{\ast},\ubar{\pi}^{\ast})$ and the solution $G$ to the dual HJB PDE (\ref{eq: dual HJB PDE general}) satisfy Condition \ref{cond: uniform integrability condition for dual}, then we directly obtain from Lemma \ref{lem: weak verification theorem dual HJB PDE}
\begin{align*}G(t,v,z)&=  \mathbbm{E}\big[U(V_{\ubar{\lambda}^{\ast}}^{v_0,\ubar{\pi}^{\ast}}(T))\  \big | \ V_{\ubar{\lambda}^{\ast}}^{v_0,\ubar{\pi}^{\ast}}(t) = v, \ z(t) =z \big] = \inf_{\lambda \in \mathcal{D}_{\text{UI}}} \sup_{\pi \in \Lambda_{\text{UI}}} \mathbbm{E}\big[U(V_{\lambda}^{v_0,\pi}(T))\  \big | \ V_{\lambda}^{v_0,\pi}(t) = v, \ z(t) =z \big]. 
\end{align*}
\end{remark} 

Note that (\ref{eq: dual optimization problem (general)}) is the minimax control problem associated with condition (C) from \cite{cvitanic1992}. However, we arrived at the same optimization problems by applying convex duality results from real constraints directly to the pointwise optimization at the level of the HJB PDE, whereas \cite{cvitanic1992} apply martingale methods to the underlying stochastic processes. \cite{cvitanic1992} go on to prove that the optimal controls for the dual control problem (\ref{eq: dual optimization problem (general)}) lead to an optimal portfolio process for the allocation constrained portfolio optimization problem $\mathbf{(P)}$. In doing so, they use the so-called Legendre-Fenchel transformation to transform the dual control problem (\ref{eq: dual optimization problem (general)}) and derive \enquote{another} dual representation of $\mathbf{(P)}$. These arguments heavily rely on the completeness of the underlying financial market and are thus not available to us. 

\subsection{Exponential Affine Separability}\label{subsec: Exponential Affine Separability}
In this section, we derive a condition under which the solution $G$ to the dual HJB PDE (\ref{eq: dual HJB PDE general}) is of an exponentially affine and separable form, i.e. 
\begin{align}\label{eq: exp affine separable value function}
G(t,v,z) = \frac{1}{b}v^b\exp \left(A(T-t)+B(T-t)'z\right),
\end{align}
for some functions $A:[0,T]\rightarrow \R,$ and $B:[0,T]\rightarrow \R^m$ with $A(0) = 0$ and $B(0) = 0$. In  a setting without the presence of allocation constraints and time-independent market coefficients, \cite{Liu2006} provides such a condition which can be directly verified for any given market coefficients (see equations (9)-(11) and (13)-(17) in \cite{Liu2006}).\footnote{In fact, the result of \cite{Liu2006} even includes the more general case of an exponentially quadratic separation. The approach we present below can be extended to include quadratic separation in a natural manner. However, such an extension would complicate the involved notation and thus diminish the presentation of the core concepts involved. Moreover, we were not able to construct realistic working examples that require quadratic separation in an allocation constrained setting. Hence, we restrict our analysis in this work to exponentially affine separation.} Under the presence of additional constraints on allocation, we need to adapt this condition suitably. \\
To this end, for any $(t,z,B)\in [0,T]\times \R^m \times \R^m$, we define $\hat{\lambda}^{\ast}(t,z,B)$ as the minimizing argument 
\begin{align}\label{eq: minimization w.r.t. lambda}
 \hat{\lambda}^{\ast}(t,z,B) &= \underset{\lambda \in \R^d}{\text{argmin}}\left \{ 2(1-b)\delta_K(\lambda) + \left \Vert \Sigma^{-1}\left( \mu - r \1 + \lambda \right) + \left(\Sigma^z\rho\right)'B \right \Vert^2 \right \}\\
 &= \underset{\lambda \in \R^d}{\text{argmin}}\left \{ 2(1-b)\delta_K(\lambda) + 2\lambda'\left(\Sigma\cdot \Sigma'\right)^{-1}\left[ \mu - r \1 + (\Sigma^z\rho\Sigma')'B\right] + \left \Vert \Sigma^{-1} \lambda \right \Vert^2\right \} \nonumber.
\end{align}
Given $\hat{\la},$ we provide a condition which ensures that (\ref{eq: exp affine separable value function}) holds. This can be achieved by considering the corresponding condition from \cite{Liu2006} and augmenting the market coefficients by $\hat{\la}.$
\begin{condition}{$(\text{EAS})$}\label{cond: Exp Separability optimal lambda*}
\textcolor{white}{1}\\
We say that Condition \ref{cond: Exp Separability optimal lambda*} is satisfied if for any $(t,z,B) \in [0,T]\times \R^m \times \R^m$ the market coefficients and the minimizer $\hat{\lambda}^{\ast}$ satisfy
\begin{align*}
	&\mu^{z}(t,z) = k_0(t) + k_1(t)z\\
	&\Sigma^z(t,z)\Sigma^z(t,z)' = h_0(t) + h_1(t)[z]\\
	&\Sigma^z(t,z)\rho(t,z)\left(\Sigma^z(t,z) \rho(t,z)\right)'-\Sigma^z(t,z)\Sigma^z(t,z)'= l_0(t) + l_1(t)[z]\\
\Leftrightarrow \ &\Sigma^z(t,z)\rho(t,z)\left(\Sigma^z(t,z) \rho(t,z)\right)' = \underbrace{\left(l_0(t)+h_0(t)\right)}_{=:\hat{l}_0(t)}+\underbrace{\left(l_1(t)+h_1(t)\right)[z]}_{=:\hat{l}_1(t)[z]}\\
	&r(t,z) + \delta_K(\hat{\lambda}^{\ast}(t,z,B)) =  p_0(t,B) + p_1(t,B)'z\\
	&\left\Vert\Sigma^{-1}(t,z)\left(\mu(t,z)+\hat{\lambda}^{\ast}(t,z,B)-r(t,z)\1\right) \right \Vert^2 = q_0(t,B) + q_1(t,B)'z \\
	&\Sigma^z(t,z)\rho(t,z)\Sigma^{-1}(t,z)\left(\mu(t,z)+\hat{\lambda}^{\ast}(t,z,B)-r(t,z)\1\right) = g_0(t,B)+g_1(t,B)z,
\end{align*}
for some functions such that $p_0(t,B),$ $q_0(t,B)\in \R$, $k_0(t),$ $p_1(t,B),$ $q_1(t,B),$ $g_0(t,B)\in \R^m$ as well as  $k_1(t),$ $h_0(t),$ $l_0(t),$ $g_1(t,B)\in \R^{m\times m}$ and the functions $h_1(t)[\cdot],$ $l_1(t)[\cdot]:\R^m \rightarrow \R^{m\times m}$ are linear\footnote{The functions $h_1[\cdot]$ and $l_1[\cdot ]$ are three-dimensional tensors, which are a generalization of vectors and matrices to higher dimensions. In our context, we may think of $h_1$ and $l_1$ as being represented by matrices, whose entries $(h_1)_{ij}$ and $(l_1)_{ij}$ are $\R^m$-valued. Upon being evaluated at a $z\in \R^m,$ each entry of $h_1[z]$ and $l_1[z]$ is obtained by computing the scalar product $z'(h_1)_{ij}$ and $z'(l_1)_{ij}$. Hence, for any $x,y\in \R^m$ and applying the rules of ordinary vector-matrix multiplication,  the products $x'h_1[\cdot]y$ and $x'l_1[\cdot]y$ are vectors in $\R^m$. In particular, $x'h_1[z]y=z'\left(x'h_1[\cdot]y\right)$ and $x'l_1[z]y=z'\left(x'l_1[\cdot]y\right)$ for any $x,y,z\in \R^m$.} for every fixed $(t,B)\in [0,T]\times \R^m.$
\end{condition}

Provided that Condition \ref{cond: Exp Separability optimal lambda*} is satisfied, we can characterize the exponents $A$ and $B$ in (\ref{eq: exp affine separable value function}) through the system of ODEs\footnote{Here, $A_{\tau}$ and $B_{\tau}$ denote the derivatives of $A$ and $B$ with respect to $\tau\in [0,T]$.}
\begin{align}\label{eq: constr ODE for A}
	A_{\tau}(\tau) &= bp_0\left(T-\tau,B(\tau)\right) + k_0(T-\tau)'B(\tau) + \frac{1}{2}B(\tau)h_0(T-\tau)B(\tau) \nonumber \\
	& \qquad + \frac{1}{2}\frac{b}{1-b}\Big[q_0\left(T-\tau,B(\tau)\right) + 2g_0\left(T-\tau,B(\tau)\right)'B(\tau) + B(\tau)\left(l_0(T-\tau)+h_0(T-\tau)\right)B(\tau)\Big]
\end{align}    
\begin{align}\label{eq: constr ODE for B}
	B_{\tau}(\tau) &= bp_1\left(T-\tau,B(\tau)\right) + k_1'(T-\tau)B(\tau) +\frac{1}{2}B(\tau)'h_1(T-\tau)[\cdot]B(\tau) \nonumber \\ 
& \qquad +\frac{1}{2}\frac{b}{1-b}\Big[q_1\left(T-\tau,B(\tau)\right) + 2g_1\left(T-\tau,B(\tau)\right)B(\tau) +  B(\tau)'\left(l_1(T-\tau)[\cdot]+h_1(T-\tau)[\cdot]\right)B(\tau)\Big].
\end{align}

\begin{theorem}\label{thm: solution to dual HJB PDE given condition EAS}
Let Condition \ref{cond: Exp Separability optimal lambda*} be satisfied and let $A,$ $B$ be solutions to the ODEs (\ref{eq: constr ODE for A}) and (\ref{eq: constr ODE for B}) with initial condition $A(0)=0,$ $B(0)= 0.$ Then,
\begin{align*}
	G(t,v,z) = \frac{1}{b}v^b\exp \left(A(T-t) + B(T-t)'z\right)
\end{align*}
is a solution to the  primal HJB PDE (\ref{eq: constr. HJB PDE Factor Model}) and the dual HJBI PDE (\ref{eq: dual HJB PDE general}).
\end{theorem}

\begin{remark}\label{rem: Solvability of ODEs}
	Although we can by no means provide explicit solutions to the ODEs (\ref{eq: constr ODE for A}) and (\ref{eq: constr ODE for B}) in general, at least the local existence of a solution is guaranteed by the existence theorems of Peano (and Picard-Lindel\"of) if their respective right-hand sides are continuous (Lipschitz-continuous). In particular, we can then obtain an approximate solution to the dual HJBI PDE (\ref{eq: dual HJB PDE general}) for small $\tau=T-t$ by approximating $A$ and $B,$ by e.g. the Euler method. 
\end{remark}

If Condition \ref{cond: Exp Separability optimal lambda*} is satisfied, then we can extend the verification approach used in \cite{Escobar2021} to formally verify the optimality of the obtained candidate optimal portfolio $\ubar{\pi}^{\ast}.$ Unlike in Lemma \ref{lem: weak verification theorem dual HJB PDE}, we only need to assume that  $\ubar{\pi}^{\ast}$ satisfies a uniform integrability condition and not make any assumption about other portfolios $\pi \in \Lambda_K.$ This is possible because we can exploit the additional knowledge that $G$ is exponentially affine due to Condition \ref{cond: Exp Separability optimal lambda*}.

\begin{theorem}[Verification Theorem Primal Problem]\label{thm: verification theorem primal OP}
	\textcolor{white}{1}\\
	Let Condition \ref{cond: Exp Separability optimal lambda*} be satisfied, let $A$ and $B$ be solutions to the ODEs (\ref{eq: constr ODE for A}) and (\ref{eq: constr ODE for B}) with initial condition $A(0)=0$ and $B(0)= 0$ and let $G$ be defined as in (\ref{eq: exp affine separable value function}). 
	Define $\ubar{\lambda}^{\ast}(t,v,z) := \hat{\lambda}^{\ast}(t,z,B(T-t))$ and 
	\begin{align}
		\ubar{\pi}^{\ast}(t,v,z) := \frac{1}{1-b}\left(\Sigma\Sigma'\right)^{-1}\left[\mu + \ubar{\lambda}^{\ast}(t,v,z) - r \1 + \left(\Sigma^z\rho\Sigma'\right)'B(T-t) \right] \label{eq: optimal constr. pi (general)}.
	\end{align}
	 If $G,$ $\ubar{\pi}^{\ast},$ $\lambda \equiv 0$ satisfy Condition \ref{cond: uniform integrability condition for dual}, then 
	 \begin{align}
	 	G(t,v,z) &= \mathbbm{E}\left[U(V^{v_0, \ubar{\pi}^{\ast}}(T))\  \big | \ V^{v_0, \ubar{\pi}^{\ast}}(t) = v, \ z(t) =z  \right] \label{eq: G attains value function for optimal strategy} \\
	 	&\geq  \mathbbm{E}\left[U(V^{v_0, \pi}(T))\  \big | \ V^{v_0, \pi}(t) = v, \ z(t) =z  \right] \quad \forall \pi \in \Lambda_K(t). \label{eq: pi star dominates all other constrained pi}
	 \end{align}
	 In particular, $G(t,v,z) = \Phi(t,v,z)$, for all $(t,v,z) = [0,T]\times (0,\infty)\times \R^m$ and $\ubar{\pi}^{\ast}$ is optimal for $\mathbf{(P)}.$
\end{theorem}


\section{Examples}\label{sec: Examples}
We consider three different choices of models for which Condition \ref{cond: Exp Separability optimal lambda*} can be verified and an explicit expression for the ODEs (\ref{eq: constr ODE for A}) and (\ref{eq: constr ODE for B}) can be derived. Throughout the examples, we always follow the same steps in chronological order:
\begin{itemize}
	\item[(i)] Define the underlying financial market model, by choosing the market coefficients $\mu^z,$ $\Sigma^z,$ $\rho,$ $r,$ $\mu$ and $\Sigma.$
	\item[(ii)] Derive an explicit representation of the minimizer $\hat{\lambda}^{\ast}$ of (\ref{eq: minimization w.r.t. lambda}) in the given market.
	\item[(iii)] Verify that Condition \ref{cond: Exp Separability optimal lambda*} is satisfied for the given market.
	\item[(iv)] Derive an explicit representation for the ODEs (\ref{eq: constr ODE for A}), (\ref{eq: constr ODE for B}) and the candidate optimal portfolio $\ubar{\pi}^{\ast}$ in terms of the market coefficients and $\hat{\lambda}^{\ast}.$
	\item[(v)] \emph{If possible:} \\
	Formally verify the optimality of $\ubar{\pi}^{\ast}$ for $\mathbf{(P)}$ by proving that $\ubar{\pi}^{\ast},$ $G$ (as in (\ref{eq: exp affine separable value function})) and $\lambda \equiv 0$ satisfy Condition \ref{cond: uniform integrability condition for dual}.
\end{itemize}

\subsection{Black-Scholes Model}\label{subsec: Black-Scholes Model}
First, we consider a $d$-dimensional Black-Scholes model $\mathcal{M}_{BS}$ with time-dependent coefficients, which is exactly the setting of Section 15 in \cite{cvitanic1992}. 

\begin{definition}[$\mathcal{M}_{BS}$]\label{def: Black-Scholes market}
	\textcolor{white}{1}\\
	Let $m=1$ and $d\in \N.\footnote{We could equivalently consider $m=0$ and completely disregard the stochastic factor $z$ and its drift, diffusion and correlation coefficients in the definition of the Black-Scholes model $\mathcal{M}_{BS}.$}$ Consider continuous functions ${r:[0,T]\rightarrow \R,}$ $\eta:[0,T]\rightarrow \R^d,$ and $\sigma:[0,T]\rightarrow \R^{d\times d}$ such that the inverse $\sigma(t)^{-1}$ exists for all $t\in [0,T].$ Then, the $d$-dimensional Black-Scholes market $\mathcal{M}_{BS}$ is defined by the market coefficients
	\begin{align*}
		&z_0 = \mu^z(t,z)=\Sigma^z(t,z)=\rho(t,z)= 0, \\
		\text{and} \quad &r(t,z) = r(t), \quad \mu(t,z) = r(t)\1+\eta(t), \quad \Sigma(t,z)=\sigma(t).
	\end{align*}
\end{definition}

We can directly apply the duality theory developed in Section \ref{sec: The Dual HJBI PDE} to verify Condition \ref{cond: Exp Separability optimal lambda*} and obtain a solution to the HJBI PDE (\ref{eq: dual HJB PDE general}) in $\mathcal{M}_{BS}.$

\begin{lemma}[Dual ODEs in $\mathcal{M}_{BS}$]\label{lem: solution to dual HJBI PDE B.S. model}
\textcolor{white}{1}\\
Consider the financial market $\mathcal{M}_{BS}.$  Then Condition \ref{cond: Exp Separability optimal lambda*} is satisfied. \\
Let 
\begin{align}\label{eq: optimal lambda B.S. market}
	\la(t)= \underset{\lambda \in \R^d}{\text{argmin}}\left \{ 2(1-b)\delta_K(\lambda) + \left \Vert \sigma(t)^{-1}(\eta(t)+\lambda) \right \Vert^2 \right \}
\end{align}
and $A:[0,T]\rightarrow \R$ satisfy $A(0)=0$ and
\begin{align*}
	A_{\tau}(\tau) = br(T-\tau) + \frac{1}{2}\frac{b}{1-b}\inf_{\lambda \in \R^d}\left \{ 2(1-b)\delta_K(\lambda) + \left \Vert \sigma(T-\tau)^{-1}(\eta(T-\tau)+\lambda) \right \Vert^2 \right \}.
\end{align*}
Then, 
$$
G(t,v,z) = \frac{1}{b}v^b\exp(A(T-t))
$$
is a solution to the dual HJBI PDE (\ref{eq: dual HJB PDE general}) and the corresponding candidate optimal portfolio is 
\begin{align*}
	\ubar{\pi}^{\ast}(t,v,z) = \frac{1}{1-b}\left(\sigma(t)\sigma(t)'\right)^{-1}\left(\eta(t)+\la(t) \right).
\end{align*}
\end{lemma}

Unsurprisingly, the candidate optimal portfolio process $\ubar{\pi}^{\ast}$ proposed by Lemma \ref{lem: solution to dual HJBI PDE B.S. model} is the same as that obtained in Example 15.2 by \cite{cvitanic1992} via the auxiliary markets methodology. Moreover, due to the simplicity of this set-up, we can even formally verify the optimality of $\ubar{\pi}^{\ast}$ by showing that Condition \ref{cond: uniform integrability condition for dual} is satisfied.

\begin{corollary}\label{cor: uniform integrability B.S. model}
	Consider the financial market $\mathcal{M}_{BS}.$ Then, $G,$ $\ubar{\pi}^{\ast}$ as in Lemma \ref{lem: solution to dual HJBI PDE B.S. model} and $\lambda \equiv 0$ satisfy Condition \ref{cond: uniform integrability condition for dual}. In particular, $\ubar{\pi}^{\ast}$ is optimal for $\mathbf{(P)}.$
\end{corollary}

\subsection{Multi-Factor Stochastic Covariance of CIR-Type}\label{subsec: Multi-Factor Stochastic Volatility of CIR Type}
Next, we consider a financial market model with a stochastic covariance matrix, which depends on $m$ independent CIR-processes. More specifically, we assume that the covariance matrix $\Sigma(t,z)$ is a block-diagonal matrix, whose diagonal blocks $\Sigma_i$ are scaled proportionally to the $i$-th CIR-process $z_i.$ One may think of the underlying financial market $\mathcal{M}_{CIR}$ as consisting of risky assets from $m$ unrelated asset classes, where the covariances within each asset class are driven by one of $m$ independent stochastic (CIR) risk factors. Special cases of this model are the Heston model (\cite{Heston1993}) for $m=d=1$ and the PCSV model with independent assets (\cite{Rubtsov2017}) for $m=d\in \N$ and $d_i=1$ for all $i=1,...,m$.\footnote{We will later see in Lemma \ref{lem: solution to dual HJBI PDE CIR model} that Condition \ref{cond: Exp Separability optimal lambda*} is only satisfied in $\mathcal{M}_{CIR}$ if the structure of the allocation constraints allows for a convenient separabiltiy in (\ref{eq: minimization w.r.t. lambda}). In the definition of $\mathcal{M}_{CIR},$ we have intentionally limited the covariance $\Sigma$ of risky assets to be of block-diagonal structure to facilitate the presentation of this fact. However, we can in principle also choose more complex models for $\Sigma,$ such as the general PCSV model, and adjust the allocation constraints accordingly without changing the underlying theory in a significant way.}\\

For notational convenience in the following discussion, we introduce the element-wise product between any two real vectors $x,y$ of identical dimension as $x\odot y$.
\begin{definition}[$\mathcal{M}_{CIR}$]\label{def: CIR market}
	\textcolor{white}{1}\\
	Let $m,d,d_1,...,d_m \in \N$ such that $m\leq d$ and $\sum_{i=1}^m d_i = d.$ Consider constants $\kappa,\theta, \sigma \in (0,\infty)^d$ such that
	\begin{align}\label{eq: Fellers Condition in m-factor CIR market}
		2\kappa_i\theta_i > \sigma_i^2 \quad \forall i=1,...,m.
	\end{align}
Moreover, let $r \in \R$ and $\rho_i\in (-1,1)^{d_i},$ $\eta_i \in \R^{d_i},$ and non-singular $\Sigma_i\in \R^{d_i\times d_i}$ be given for $i=1,...,m.$ Then, the $d$-dimensional market $\mathcal{M}_{CIR}$ with $m$-factor volatility of CIR-type is defined by the market coefficients
\begin{align*}
	&\mu^z(t,z) = \kappa \odot \left(\theta - z\right), \quad \Sigma^z(t,z) = 
	\begin{pmatrix} \sigma_1 \sqrt{z_1} & & 0\\
		& \ddots & \\
		 0& & \sigma_m \sqrt{z_m}
	\end{pmatrix}, \quad \rho(t,z) = \begin{pmatrix}
	\rho_1' & & 0 \\
	& \ddots & \\
	0 & & \rho_m'
\end{pmatrix}\in \R^{m \times d}, \\
& r(t,z) = r, \quad \mu(t,z) = r(t,z)\1 + \begin{pmatrix}
	\eta_1 z_1 \\
	\vdots\\
	\eta_m z_m
\end{pmatrix}, 
\quad \Sigma(t,z) = \begin{pmatrix}
\Sigma_1 \sqrt{z_1} & & 0 \\
& \ddots & \\
0 & & \Sigma_m \sqrt{z_m}
\end{pmatrix}.
\end{align*}
\end{definition}
In $\mathcal{M}_{CIR},$ the minimization (\ref{eq: minimization w.r.t. lambda}) can be equivalently rewritten as\footnote{Compare to the derivation of (\ref{eq: optimizer in CIR market}) in the proof of the subsequent Lemma \ref{lem: solution to dual HJBI PDE CIR model} for details.}
\begin{align*}
	 \underset{\substack{\lambda = \left(\lambda_1,...,\lambda_m\right)' \\	\lambda_i \in \R^{d_i}}}{\text{argmin}}\left \{ 2(1-b)\delta_{K}\left(\lambda\right) + \sum_{i=1}^m \left( 2\left(\Sigma_i^{-1}\lambda_i\right)'\left(\Sigma_i^{-1}\eta_i + \sigma_i B_i \rho_i \right) + \left \Vert \Sigma_i^{-1}\lambda_i \right \Vert^2 z_i \right)\right \}.
\end{align*}
However, as the underlying financial market model $\mathcal{M}_{CIR}$ consists of $m$ independent asset classes, it is natural to assume a certain independence with respect to the allocation constraints, too. This independence can be expressed in mathematical terms by assuming that $K$ can be written as the Cartesian product of $m$ constraints $K_1,...,K_m$ on the individual asset classes.
\begin{lemma}[Dual ODEs in $\mathcal{M}_{CIR}$]\label{lem: solution to dual HJBI PDE CIR model}
	\textcolor{white}{1}\\
	Consider the financial market $\mathcal{M}_{CIR}.$ If $K=\bigtimes_{i=1}^m K_i$ with $K_i \subset \R^{d_i}$ closed convex and non-empty interior for every $i=1,..,m,$ then Condition \ref{cond: Exp Separability optimal lambda*} is satisfied.\\
	Let
	\begin{align*}
		\la(t,z,B):= \begin{pmatrix}
			\la_1(B_1)z_1 \\
			\vdots \\
			\la_m(B_m)z_m
		\end{pmatrix},
	\end{align*}
	where
	\begin{align*}
		\la_i(B_i)=\underset{\lambda_i \in \R^{d_i}}{\text{argmin}}\left\{2(1-b)\delta_{K_i}\left(\lambda_i\right) + \left \Vert   \Sigma_i^{-1}\left(\eta_i + \lambda_i\right) +  \sigma_i B_i \rho_i \right \Vert^2 \right \}
	\end{align*}
	and $A:[0,T]\rightarrow \R,$ $B:[0,T]\rightarrow \R^m$ satisfy $A(0)=0,$ $B(0)=0$ and
	\begin{align}
		A_{\tau}(\tau) &= br + \left(\kappa \odot \theta \right)'B(\tau)\nonumber \\
		\left(B_{\tau}\right)_i(\tau)&=-\kappa_i B_i(\tau) + \frac{1}{2}\sigma_i^2\left(B_i(\tau)\right)^2  +\frac{1}{2}\frac{b}{1-b}\inf_{\lambda_i \in \R^{d_i}}\left\{2(1-b)\delta_{K_i}\left(\lambda_i\right) + \left \Vert   \Sigma_i^{-1}\left(\eta_i + \lambda_i\right) +  \sigma_i B_i \rho_i \right \Vert^2 \right \}. \nonumber
	\end{align}
	Then, 
	$$
	G(t,v,z) = \frac{1}{b}v^b\exp(A(T-t)+B(T-t)'z)
	$$
	is a solution to the dual HJBI PDE (\ref{eq: dual HJB PDE general}) and the corresponding candidate optimal portfolio is
	\begin{align*}
		\ubar{\pi}^{\ast}(t,v,z) = \begin{pmatrix}
			\pi^{\ast}_1(B_1(T-t))\\
			\vdots\\
			\pi^{\ast}_m(B_m(T-t))
		\end{pmatrix},
	\end{align*}
with
	\begin{align*}
	 \pi^{\ast}_i(B_i(T-t)) = \frac{1}{1-b}\left(\Sigma_i\Sigma_i'\right)^{-1}\Big(\eta_i + \la_i(B_i(T-t)) + \sigma_iB_i(T-t)\Sigma_i\rho_i\Big).
	\end{align*}
\end{lemma}

The ODEs for $B_i$ in Lemma \ref{lem: solution to dual HJBI PDE CIR model} do not admit a general closed-form solution, as the right-hand side of the ODE still depends on a convex minimization problem. Note however, as long as all minimizer $\la_i(B_i)$ are continuous in $B_i$ (e.g. if each $K_i$ are compact sets), then the right-hand side of each ODE for $B_i$ is continuous in $B_i$ and therefore admits at least a local solution.\\
Due to this lack of an explicit representation for $B,$ we also lack an explicit representation for $\ubar{\pi}^{\ast}.$  In addition, the quadratic variations of $\ln\left(V^{v_0, \ubar{\pi}^{\ast}}\right)$ and $z$ are stochastic and we can thus no longer follow the approach from Corollary \ref{cor: uniform integrability B.S. model} to formally verify the optimality of $\ubar{\pi}^{\ast}.$ A thorough investigation of this topic is beyond the scope of this paper and will therefore be left for future research.  

\subsection{Multi-Factor Short Rate of OU-Type}\label{subsec: Multi-Factor Short Rate of OU-Type}
Lastly, we consider a financial market $\mathcal{M}_{OU}$ with a stochastic short rate $r,$ which is driven by an $m$-dimensional Ornstein-Uhlenbeck process and $d=m$ zero-coupon bonds with maturities $T_1,...,T_m > T$ as primary traded assets (similar models were discussed e.g. in Section 7.3 in \cite{Kwok2008} for derivatives pricing, in \cite{Zagst2007} for economic scenario generation as well as in \cite{Sorensen1999} and \cite{Liu2006} in a portfolio optimization context). For this purpose, we define $\mu^z$ and $\Sigma^z$ as
\begin{align}\label{eq: Ohrnstein-Uhlenbeck dynamics z}
	\mu^z(t,z) = \kappa\odot \left[\theta-z\right], \quad \Sigma^z(t,z) = \sigma,
\end{align}
for arbitrary constants $\kappa \in (0,\infty)^m$, $\theta \in \R^m,$ and a non-singular matrix $\sigma \in \R^{m\times m}.$ For two weights $w_0 \in \R,$ $w_1 \in \R^m$ we then define the short rate $r$ through
\begin{align*}
	r(t,z) = w_0 + w_1'z.
\end{align*}
In particular, the $Q$-dynamics of the short rate are given as 
\begin{align*}
dr(t,z(t)) &= d(w_0 + w_1'z(t)) = w_1'dz(t) = w_1'\left(\kappa \odot \left[\theta -z(t)\right] \right)dt + w_1'\sigma dW^z(t) \\
\end{align*}
To formally define $\mathcal{M}_{OU},$ we still need to explicitly determine the dynamics of the traded zero-coupon bonds. We determine these dynamics via risk-neutral pricing. Assuming a constant market price of risk $\eta \in \R^m,$ we can define the equivalent martingale measure $\tilde{Q}$ through its Radon-Nikodym derivative
\begin{align*}
	\frac{d\tilde{Q}}{dQ} = \exp\left( -\frac{1}{2} \Vert \eta \Vert^2T - \int_0^T\eta'dW^z(t)\right).
\end{align*}
Then, according to Girsanov's theorem, there exists $\tilde{Q}$-Wiener process $\tilde{W}^z$ such that the $\tilde{Q}$-dyamics of the short rate are given as 
\begin{align*}
	dr(t,z(t)) = \left(\kappa \odot \left[\theta -z(t)\right] - w_1'\sigma\eta \right)dt + w_1'\sigma d\tilde{W}^z(t).
\end{align*}
Moreover, we can now make use of risk-neutral pricing to determine the arbitrage-free prices of zero-coupon bonds with different maturities. The financial market $\mathcal{M}_{OU}$ belongs to the group of affine factor models (without stochastic volatility) studied in \cite{Duffie1996}. Hence, there exist suitable deterministic continuously differentiable functions $a:(0,\infty)\rightarrow \R,$ $b:(0,\infty)\rightarrow \R^m$ such that the price of a zero-coupon bond with maturity $T_i$ at time $t\in[0,T].$
\begin{align}\label{eq: ZCB Exponential Affine}
	P(t,T_i) &= \E_{\tilde{Q}}\left[\exp\left(-\int_t^{T_i}r(s,z(s))ds\right) \ \Big | \mathcal{F}_t\right] = \exp\left(a(T_i-t) + b(T_i-t)'z(t)\right).
\end{align}
 By applying It\^{o}'s formula and noting that the discounted price processes \linebreak $\left(\nicefrac{P(t,T_i)}{P_0(t)}\right)_{t \in [0,T]}$ are martingales with respect to $\tilde{Q},$ we see that the $Q$-dynamics of $\left(P(t,T_i)\right)_{t\in [0,T]}$ are 
\begin{align*}
	dP(t,T_i) = P(t,T_i)\left(\left[r(t,z(t)) + b(T_i-t)'\sigma \eta \right]dt +  b(T_i-t)'\sigma dW^z(t)\right).
\end{align*}
These zero-coupon bonds $P(t,T_1),...,P(t,T_m)$ constitute the primary traded assets of the financial market $\mathcal{M}_{OU}$ which is formally defined below.

\begin{definition}[$\mathcal{M}_{OU}$]\label{def: OU-Market}
	\textcolor{white}{1}\\
	Let $m=d\in \N.$ Consider constants $w_0\in \R,$ $\kappa \in (0,\infty)^m,$ $w_1, \theta, \eta \in \R^m,$ a non-singular matrix $\sigma \in \R^{m\times m},$ maturities $\hat{T}=(T_1,...,T_m)'\in (T,\infty)^m$ and a continuously differentiable function $b:(0,\infty)\rightarrow \R^m$ such that the matrix
	$$
	b(t;\hat{T})=\left(b(T_1-t),...,b(T_m-t)\right)\in \R^{m\times m}
	$$
	has an inverse $b(t;\hat{T})^{-1}$ for all $t\in [0,T].$ Then, the $m$-dimensional Bond market $\mathcal{M}_{OU}$ with OU short rate is defined by the market coefficients
	\begin{align*}
		&\mu^z(t,z) = \kappa \odot\left(\theta-z\right), \quad \Sigma^z(t,z) = \sigma, \quad \rho(t,z) = I_m \\
		&r(t,z) = w_0 + w_1'z, \quad \mu(t,z) = r(t,z)\1 + b(t;\hat{T})'\sigma \eta, \quad \Sigma(t,z) = b(t;\hat{T})'\sigma.
	\end{align*}
\end{definition}
Despite the stochastic short rate, the market $\mathcal{M}_{OU}$ is surprisingly tractable. Specifically, none of the terms involved in the minimization (\ref{eq: minimization w.r.t. lambda}) are dependent on the stochastic factor $z,$ which results in a time-dependent but deterministic candidate optimal portfolio $\ubar{\pi}^{\ast}.$
\begin{lemma}[Dual ODEs in $\mathcal{M}_{OU}$]\label{lem: solution to dual HJBI PDE OU model}
	\textcolor{white}{1}\\
	Consider the financial market $\mathcal{M}_{OU}.$ Then Condition \ref{cond: Exp Separability optimal lambda*} is satisfied. \\
	Let 
	\begin{align*}
		\la(t,B)= \underset{\lambda \in \R^d}{\text{argmin}}\left \{ 2(1-b)\delta_K(\lambda) + \left \Vert \eta + \left(b(t;\hat{T})'\sigma\right)^{-1}\lambda + \sigma'B \right \Vert^2 \right \}
	\end{align*}
	and $A:[0,T]\rightarrow \R,$ $B:[0,T]\rightarrow \R^m$ satisfy $A(0)=0,$ $B(0)=0$ and
	\begin{align}
		A_{\tau}(\tau) &= bw_0+\left(\kappa \odot \theta \right)'B(\tau) + \frac{1}{2}\Vert \sigma'B(\tau) \Vert^2  \nonumber \\
		&\quad  + \frac{1}{2}\frac{b}{1-b}\inf_{\lambda \in \R^d}\left \{ 2(1-b)\delta_K(\lambda) + \left \Vert \eta + \left(b(T-\tau;\hat{T})'\sigma\right)^{-1}\lambda + \sigma'B \right \Vert^2 \right \} \nonumber \\
		B_{\tau}(\tau)&= bw_1 - \kappa \odot B(\tau) \label{eq: ODE for B OU Model}.
	\end{align}
	Then, 
	$$
	G(t,v,z) = \frac{1}{b}v^b\exp(A(T-t)+B(T-t)'z)
	$$
	is a solution to the dual HJBI PDE (\ref{eq: dual HJB PDE general}) and the corresponding candidate optimal portfolio $\ubar{\pi}^{\ast}$ (as in (\ref{eq: optimal constr. pi (general)})) is given as
	\begin{align*}
		\ubar{\pi}^{\ast}(t,v,z) = \frac{1}{1-b}\left(\sigma' b(t;\hat{T})\right)^{-1}\Big(\eta + \left(b(t;\hat{T})'\sigma\right)^{-1}\la(t,B(T-t)) + \sigma'B(T-t)\Big).
	\end{align*}
\end{lemma}

\begin{remark}\label{rem: explicit solution B OU Model}
	The solution to the ODE (\ref{eq: ODE for B OU Model}) is known in closed-form (see e.g. equations (6) and (7) in Chapter 1, Section \S 2 of \cite{Walter1998}) and is given by $B(\tau)= (B_1(\tau),...,B_m(\tau))'$ with 
	$$
	B_i(\tau) = \left(w_1\right)_i b e^{-\kappa_i \tau}\int_0^\tau e^{\kappa_i s}ds  = \frac{\left(w_1\right)_ib}{\kappa_i}\left(1-e^{-\kappa_i \tau}\right).
	$$
\end{remark}
Unlike in $\mathcal{M}_{CIR},$ the quadratic variations of $\ln\left(V^{v_0, \ubar{\pi}^{\ast}}\right)$ and $z$ are deterministic and bounded in $\mathcal{M}_{OU}.$ Therefore, it is straightforward to adapt the proof of Corollary \ref{cor: uniform integrability B.S. model} to formally verify the optimality of $\ubar{\pi}^{\ast}$ for $\mathbf{(P)}.$
\begin{corollary}\label{cor: uniform integrability OU model}
	Consider the financial market $\mathcal{M}_{OU}.$ Then, $G,$ $\ubar{\pi}^{\ast}$ as in Lemma \ref{lem: solution to dual HJBI PDE OU model} and $\lambda \equiv 0$ satisfy Condition \ref{cond: uniform integrability condition for dual}. In particular, $\ubar{\pi}^{\ast}$ is optimal for $\mathbf{(P)}.$
\end{corollary}

\section{Conclusion}\label{sec: conclusion}
In this paper, we examined a portfolio optimization problem in a financial market where asset dynamics depend on a stochastic factor. In the spirit of \cite{Liu2006}, we were able to derive Condition \ref{cond: Exp Separability optimal lambda*} which guarantees that the solution to the HJB PDE for an allocation constrained portfolio optimization problem is exponentially affine and separable in wealth and the stochastic factor. We were able to use Condition \ref{cond: Exp Separability optimal lambda*} to characterize the optimal allocation constrained portfolio up to the solution of a deterministic optimization problem and the solution of Riccati ODEs in a market with stochastic volatility of CIR-type and in a market with stochastic short rate of OU-type. Special examples of these models include the Heston model, the PCSV model and the Vasicek model. We derived a formal verification result for the market with stochastic short rate and a general verification result up to a uniform integrability condition for arbitrary markets satisfying Condition \ref{cond: Exp Separability optimal lambda*}. The proposed methodology is general enough to derive and study the optimal allocation constrained portfolios for several financial markets with dynamics depending on a stochastic factor. Moreover, it would be interesting to continue investigating the case of CIR-type  and related stochastic volatility models, which are extremely common in the mathematical finance literature. Areas worth investigating include the solvability of the related Riccati ODEs, properties of the optimal portfolios as well as the formal verification of their optimality. We leave such analysis for future research.

\section*{Data Availability Statement}
Data sharing not applicable to this article as no datasets were generated or analysed during the current study.

\printbibliography

\newpage

\section*{Appendix: Proofs}\label{sec: appendix proofs}

\begin{proof}[Proof of Lemma \ref{lem: saddle point implies strong duality}]\textcolor{white}{123}\\
\enquote{(i) $\Rightarrow$ (ii)}: 
Let (i) hold.  Then, 
\begin{align*}
    \Phi_P &\hspace{10.5mm}=\hspace{10.5mm} \sup_{x\in \R^d}f(x) \overset{x^{\ast} \ \text{optimal for} \ \mathbf{(P)}}{=} f(x^{\ast}) = F(x^{\ast},0) = F^{\ast \ast}(x^{\ast},0)\\
   & \hspace{10.5mm} = \hspace{10.5mm}\inf_{\lambda \in \R^d}\big(\lambda \cdot 0 + \underbrace{F^{\ast}(x^{\ast},\lambda)}_{L(x^{\ast},\lambda)}\big) = \inf_{\lambda \in \R^d}L(x^{\ast},\lambda) 
    \leq L(x^{\ast},\lambda^{\ast}) \leq \sup_{x\in \R^d} L(x,\la) \\
    &\overset{\la \ \text{optimal for} \ \mathbf{(D)}}{=} \inf_{\lambda \in \R^d}\sup_{x\in \R^d}L(x,\lambda) = \Psi_D = \Phi_P.
\end{align*}
This yields
$$\sup_{x\in \R^d}L(x,\la) = L(x^{\ast},\la) = \inf_{\lambda \in \R^d}L(x^{\ast},\lambda)$$ 
and therefore $(x^{\ast},\la)$ is a saddle-point. \\

\enquote{(ii) $\Rightarrow$ (i)}: Let (ii) hold.   As $(x^{\ast},\lambda^{\ast})$ is a saddle-point of $L$, we have 
$$\sup_{x\in \R^d} L(x,\lambda^{\ast}) = L(x^{\ast},\lambda^{\ast}) \quad \text{and} \quad \inf_{\lambda \in \R^d}L(x^{\ast},\lambda) = L(x^{\ast},\lambda^{\ast}).$$
However, as we additionally have
\begin{align*}
L(x^{\ast},\lambda^{\ast}) = \sup_{x\in \R^d}L(x,\lambda^{\ast}) \geq \sup_{x\in \R^d}\inf_{\lambda \in \R^d} L(x,\lambda) \geq \inf_{\lambda \in \R^d} L(x^{\ast},\lambda) = L(x^{\ast},\lambda^{\ast})
\end{align*}
and
\begin{align*}
    L(x^{\ast},\lambda^{\ast}) = \inf_{\lambda \in \R^d} L(x^{\ast},\lambda) \leq \inf_{\lambda \in \R^d} \sup_{x\in \R^d} L(x,\lambda) \leq  \sup_{x\in \R^d}L(x,\lambda^{\ast}) = L(x^{\ast},\lambda^{\ast}),
\end{align*}
we obtain
\begin{align}\label{eq: right side of strong duality}
    \Psi_D = \inf_{\lambda \in \R^d} \sup_{x\in \R^d} L(x,\lambda) = \sup_{x\in \R^d} L(x,\lambda^{\ast}),
\end{align}
i.e. $\lambda^{\ast}$ is optimal for $\mathbf{(D)}$. Moreover, as $F(x^{\ast},\cdot)$ is u.s.c. and concave in $u$ by assumption, we get
\begin{align}\label{eq: left side of strong duality}
    \sup_{x\in \R^d} f(x) &= \sup_{x\in \R^d}F(x,0) = \sup_{x\in \R} F^{\ast \ast}(x,0) = \sup_{x\in \R^d} \inf_{\lambda \in \R^d}\big(\lambda ' 0 + F^{\ast}(x,\lambda)\big) \nonumber \\
    &= \sup_{x\in \R^d} \inf_{\lambda \in \R^d}L(x,\lambda) = \inf_{\lambda \in \R^d} L(x^{\ast},\lambda) =\inf_{\lambda \in \R^d}\big(\lambda'0 + F^{\ast}(x^{\ast},\lambda)\big)   \nonumber \\
    &= F^{\ast \ast}(x^{\ast},0) = F(x^{\ast},0) = f(x^{\ast}), 
\end{align}
i.e. $x^{\ast}$ is optimal for $\mathbf{(P)}$. Finally, this yields
\begin{align*} 
\Phi_P &= \sup_{x\in \R^d} \inf_{ \lambda \in \R^d}L(x,\lambda) \overset{(\ref{eq: left side of strong duality})}{=}\inf_{\lambda \in \R^d} L(x^{\ast},\lambda) =L(x^{\ast},\lambda^{\ast}) = \sup_{x\in \R^d} L(x,\lambda^{\ast})\\
&\overset{(\ref{eq: right side of strong duality})}{=}\inf_{\lambda  \in \R^d} \sup_{x \in \R^d} L(x,\lambda) = \Psi_D.
\end{align*}
\end{proof}
\vspace{0.5cm}
\begin{proof}[Proof of Lemma \ref{lem: Slater's condition}]
	Follows immediately by considering $-F$ in Theorem 18a) in \cite{Rockafellar1974}.
\end{proof}
\vspace{0.5cm}
\begin{proof}[Proof of Lemma \ref{lem: constrained real optimization}]
	Define $\Chi_K:\R^d\rightarrow \{0,-\infty\}$ as $\Chi_K(x) = 0$ if $x\in K$ and $\Chi_K(x) = -\infty$ if $x\notin K$. Further, set $f(x) = \tilde{f}(x) + \Chi_K(x)$ and  $F(x,u) = \tilde{f}(x) +  \Chi_K(x+u)$, which yields
	\begin{align*}
		\sup_{x\in K}\tilde{f}(x) = \sup_{x\in \R^d}f(x) \overset{F(x,0) = f(x)}{=} \sup_{x\in \R^d}F(x,0),
	\end{align*}
	i.e. $F$ is a pertubation of $f$. The corresponding Lagrangian $L$ can be computed as
	\begin{align*}
		L(x,\lambda) &= F^{\ast}(x,\lambda) = \sup_{u \in \R^d} \big( F(x,u) - u' \lambda\big) \\
		&= \sup_{u \in \R^d} \big( \tilde{f}(x) + \Chi_K(x+u)- u'\lambda\big) \\
		&= \tilde{f}(x) + \sup_{u \in \R^d} \big(\Chi_K(\underbrace{x+u}_{=v \in \R^d}) - (\underbrace{x+ u}_{=v \in \R^d})'\lambda  +  x'  \lambda \big) \\
		&= \tilde{f}(x) +  x' \lambda  + \sup_{v \in \R^d} \big(\Chi_K(v) - v' \lambda  \big) \\
		&= \tilde{f}(x) +  x'\lambda  + \underbrace{\sup_{v \in K} \big(- v' \lambda  \big)}_{= \delta_K(\lambda)} \\
		&= \tilde{f}(x) +  x' \lambda + \delta_K(\lambda).
	\end{align*}
	
	We verify that  $F$ satisfies Slater's condition from Lemma \ref{lem: Slater's condition}: \\
	By assumption, $\tilde{f}$ is concave and u.s.c.. The constraint set $K\subset \R^d$ is closed-convex and thus $\Chi_K(x+u)$ is u.s.c. and concave in $(x,u)$. Hence, $F(x,u) = \tilde{f}(x)+\Chi_K(x+u)$ is u.s.c. and concave in $(x,u)$. \\
	Further, $f$ is continuous and there exists $\hat{x}\in int \ K$. Thus, there exists a bounded environment $B_{\epsilon}(0)$ around $\bar{u}=0$ such that $\hat{x}+u \in K$ for all $u\in B_{\epsilon}(0)$. Hence, $F(\hat{x},u)= \tilde{f}(\hat{x})$ for all $u\in B_{\epsilon}(0)$, i.e. $F(\hat{x},\cdot)$ is bounded from  below by $\tilde{f}(\hat{x})$ on $ B_{\epsilon}(0)$.\\
	
	Thus, Slater's condition is satisfied and we have by Lemma \ref{lem: Slater's condition}
	\begin{align*}
		\sup_{x\in K}\tilde{f}(x) &= \sup_{x\in \R^d}f(x) = \Phi_P = \Psi_D = \inf_{\lambda \in \R^d}\sup_{x \in R^d}L(x,\lambda) = \inf_{\lambda \in \R^d}\sup_{x \in \R^d} \tilde{f}(x) +  x' \lambda + \delta_K(\lambda).
	\end{align*}
	Moreover, let now $\la \in \R^d$ be the minimizer of $\sup_{x\in \R^d}L(x,\lambda)$. By assumption, $\tilde{f}$ has a unique maximizer $x^{\ast}$ over $K$, which is simultaneously the unique global maximizer of $f$ on $K$. Similarly, the Lagrangian $L(\la,x) = \tilde{f}(x) + x'\la 
	+ \delta_K(\la)$ admits a unique maximizer $\hat{x}$. Thus, as both maximizers $x^{\ast}$ and $\hat{x}$ are unique, Lemma \ref{lem: saddle point implies strong duality} implies that they are equal. Hence, (\ref{eq: optimal constrained x as unconstrained maximizer of general dual 1}) and (\ref{eq: optimal constrained x as unconstrained maximizer of general dual 2}) are equivalent.
\end{proof}	
\vspace{0.5cm}
\begin{proof}[Proof of Lemma \ref{lem: dual HJB PDE}]
From Lemma \ref{lem: constrained real optimization} and Remark \ref{rem: constrained quadratic optimization} it is known that 
\begin{align*}
	0= G_t &+ v r G_v + \left(\mu^z\right)'\left(\nabla_z G\right) + \frac{1}{2}\text{Trace}\left[\Sigma^z\left( \Sigma^z\right) '\nabla^2_zG\right] \nonumber \\
	& + v \sup_{\pi \in K}\Big \{ (\mu  - r \1)'\pi G_v  + \left(\Sigma^z \rho \Sigma'\pi\right)' \nabla_z \left(G_v\right) +  \frac{1}{2}v \Vert \Sigma'\pi \Vert^2 G_{vv} \Big \} \\
	= G_t &+ v r G_v + \left(\mu^z\right)'\left(\nabla_z G\right) + \frac{1}{2}\text{Trace}\left[\Sigma^z\left( \Sigma^z\right) '\nabla^2_zG\right] \nonumber \\
	& + v \inf_{\lambda \in \R^d}\sup_{\pi \in \R^d}\Big \{\delta_K(\lambda)+\lambda'\pi +(\mu  - r \1)'\pi G_v   +  \left(\Sigma^z \rho \Sigma'\pi\right)' \nabla_z \left(G_v\right) +  \frac{1}{2}v \Vert \Sigma'\pi \Vert^2 G_{vv} \Big \} \\
	= G_t &+ v r G_v + \left(\mu^z\right)'\left(\nabla_z G\right) + \frac{1}{2}\text{Trace}\left[\Sigma^z\left( \Sigma^z\right) '\nabla^2_zG\right] \nonumber \\
	& + v \inf_{\lambda \in \R^d}\sup_{\pi \in \R^d}\Big \{\delta_K(\underbrace{G_v}_{>0}\cdot\lambda)+\underbrace{G_v}_{>0}\cdot\lambda'\pi +(\mu  - r \1)'\pi G_v +  \left(\Sigma^z \rho \Sigma'\pi\right)' \nabla_z \left(G_v\right) +  \frac{1}{2}v \Vert \Sigma'\pi \Vert^2 G_{vv} \Big \}
\end{align*}
By assumption, $v$ and $G_v(t,v,z)$ are positive\footnote{Note that this assumption is valid due to the exponential structure of the wealth process $V^{v_0,\pi}$ and the strict monotonicity of the utility function $U$.}. Using that $\delta_K$ is a support function and thus positive homogenous of order 1, yields
\begin{align}
0= G_t &+ v r G_v + \left(\mu^z\right)'\left(\nabla_z G\right) + \frac{1}{2}\text{Trace}\left[\Sigma^z\left( \Sigma^z\right) '\nabla^2_zG\right] \nonumber \\
& + v \inf_{\lambda \in \R^d}\sup_{\pi \in \R^d}\Big \{ \left[ \delta_K(\lambda) +(\mu + \lambda - r \1)'\pi \right]G_v   +  \left(\Sigma^z \rho \Sigma'\pi\right)' \nabla_z \left(G_v\right) +  \frac{1}{2}v \Vert \Sigma'\pi \Vert^2 G_{vv} \Big \}. \label{eq: dual HJBI before plugging in optimal pi}
\end{align}

Further, for every fixed $\lambda \in \R^d,$ the first-order optimality condition for the maximization over $\pi \in \R^d$ yields the candidate optimizer $\pi_{\lambda}$ through
\begin{align*}
	0&\overset{!}{=} \left[\mu + \lambda - r\1\right]G_v + \left( \Sigma^z \rho \Sigma'\right)'\nabla_z\left(G_v\right) + v\Sigma \Sigma'\pi_{\lambda} G_{vv} \\
	\Leftrightarrow \quad \pi_{\lambda}&= -\frac{1}{vG_{vv}}\left(\Sigma \Sigma'\right)^{-1}\left( \left[\mu + \lambda - r\1\right]G_v + \left( \Sigma^z \rho \Sigma'\right)'\nabla_z\left(G_v\right)\right).
\end{align*}
Note that
\begin{align*}
\left[\mu + \lambda - r\1\right]'\pi_{\lambda}G_v +  \left(\Sigma^z \rho \Sigma'\pi_{\lambda}\right)' \nabla_z \left(G_v\right) &=\pi_{\lambda}'\left[\left(\mu + \lambda - r\1\right)G_v + \left(\Sigma^z\rho\Sigma'\right)'\nabla_z\left(G_v\right)\right]\\
&=\pi_{\lambda}'\left[-vG_{vv}\Sigma\Sigma'\pi_{\lambda}\right] \\
&\quad = -v\Vert \Sigma' \pi_{\lambda}\Vert^2G_{vv}.
\end{align*}
Since we are maximizing a quadratic function with respect to $\pi,$ the first-order optimality condition is both necessary and sufficient. Hence, plugging $\pi_{\lambda}$  into (\ref{eq: dual HJBI before plugging in optimal pi}) finally yields
\begin{align*}
0= G_t &+ v r G_v + \left(\mu^z\right)'\left(\nabla_z G\right) + \frac{1}{2}\text{Trace}\left[\Sigma^z\left( \Sigma^z\right) '\nabla^2_zG\right] + v \inf_{\lambda \in \R^d}\Big \{ \delta_K(\lambda) -\frac{1}{2}v \Vert \Sigma'\pi_{\lambda} \Vert^2 G_{vv} \Big \} \\
= G_t &+ v r G_v + \left(\mu^z\right)'\left(\nabla_z G\right) + \frac{1}{2}\text{Trace}\left[\Sigma^z\left( \Sigma^z\right) '\nabla^2_zG\right] \nonumber \\
& + v \inf_{\lambda \in \R^d}\Big \{ \delta_K (\lambda) G_v- \frac{1}{2}\frac{1}{vG_{vv}}\Vert \Sigma^{-1}\left[\mu + \lambda - r\1\right]G_v + \left(\Sigma^z\rho\right)'\nabla_z\left(G_v\right)\Vert^2 \Big \}.
\end{align*}
\end{proof}
\vspace{0.5cm}
\begin{proof}[Proof of Lemma \ref{lem: weak verification theorem dual HJB PDE}]
For any $\pi \in \Lambda$ ($\ubar{\pi}\in \ubar{\Lambda}$) and $\lambda \in \mathcal{D}$ ($\ubar{\lambda}\in \ubar{\mathcal{D}}$) define the operator
\begin{align*}
    \mathcal{H}^{\pi}_{\lambda}G(t,v,z) &= G_t + v r G_v + \left(\mu^z\right)'\left(\nabla_z G\right) + \frac{1}{2}\text{Trace}\left[\Sigma^z\left( \Sigma^z\right) '\nabla^2_zG\right] \\
    & \quad + v \Big \{ \left[ \delta_K(\lambda) +(\mu + \lambda - r \1)'\pi \right]G_v   +  \left(\Sigma^z \rho \Sigma'\pi\right)' \nabla_z \left(G_v\right) +  \frac{1}{2}v \Vert \Sigma'\pi \Vert^2 G_{vv} \Big \},
\end{align*}
Let $(\ubar{\lambda}^{\ast},\pi)\in \ubar{\mathcal{D}}\times \Lambda$ satisfy Condition \ref{cond: uniform integrability condition for dual}. Then,  by It\^{o}'s lemma and the boundedness of the integrands, due to the definition of $\tau^{\ubar{\lambda}^{\ast}}_{n,t}$, we have
\begin{align*}
    \E\Big[G\left(\tau^{\ubar{\lambda}^{\ast}}_{n,t}, V_{\ubar{\lambda}^{\ast}}^{v_0, \pi}(\tau^{\ubar{\lambda}^{\ast}}_{n,t}),z(\tau^{\ubar{\lambda}^{\ast}}_{n,t})\right)\ \Big | \ \mathcal{F}_t\Big]    &= G\left(t,V_{\ubar{\lambda}^{\ast}}^{v_0, \pi}(t), z(t)\right) + \E\Big[\int_t^{\tau^{\ubar{\lambda}^{\ast}}_{n,t}} \underbrace{\mathcal{H}^{\pi}_{\ubar{\lambda}^{\ast}}G\left(s,V_{\ubar{\lambda}^{\ast}}^{v_0, \pi}(s), z(s)\right)}_{\leq \mathcal{H}^{\ubar{\pi}^{\ast}}_{\ubar{\lambda}^{\ast}}G\left(s,V_{\ubar{\lambda}^{\ast}}^{v_0, \ubar{\pi}^{\ast}}(s), z(s)\right) = 0}ds\ \Big | \ \mathcal{F}_t\Big] \\
    &\quad + \underbrace{\E\Big[\int_t^{\tau^{\ubar{\lambda}^{\ast}}_{n,t}} V_{\ubar{\lambda}^{\ast}}^{v_0, \pi}(s)G_v\left(s,V^{v_0, \pi}_{\ubar{\lambda}^{\ast}}(s),z(s)\right)\cdot \pi(s)'\Sigma(s,z(s)) dW(s)\ \Big | \ \mathcal{F}_t\Big]}_{= 0}\\
    & \quad + \underbrace{\E\Big[ \int_t^{\tau^{\ubar{\lambda}^{\ast}}_{n,t}} \nabla_z\left(G\right)\left(s,V_{\ubar{\lambda}^{\ast}}^{v_0, \pi}(s),z(s)\right)'\Sigma^z(s,z(s))dW^z(s)\ \Big | \ \mathcal{F}_t\Big]}_{= 0} \\
    &\leq G(t,V_{\ubar{\lambda}^{\ast}}^{v_0, \pi}(t), z(t))
\end{align*}
Taking conditional expectations and the limit $n\rightarrow \infty$ on both sides as well as recalling Remark \ref{rem: consequance condition UI} yields
\begin{align*}
      \mathbbm{E}\Big[U(V_{\ubar{\lambda}^{\ast}}^{v_0, \pi}(T))\ \Big | \ \mathcal{F}_t \Big] &\hspace{1.5 mm} =  \mathbbm{E}\Big[G(T,V_{\ubar{\lambda}^{\ast}}^{v_0, \pi}(T),z(T))\ \Big | \ \mathcal{F}_t \Big] \\
    & \hspace{1.5 mm} = \mathbbm{E}\Big[\lim_{n\rightarrow \infty}G(\tau^{\ubar{\lambda}^{\ast}}_{n,t},V_{\ubar{\lambda}^{\ast}}^{v_0, \pi}(\tau^{\ubar{\lambda}^{\ast}}_{n,t}),z(\tau^{\ubar{\lambda}^{\ast}}_{n,t}))\ \Big | \ \mathcal{F}_t \Big] \\
    &\overset{\text{\ref{cond: uniform integrability condition for dual}}}{=} 
    \lim_{n\rightarrow \infty} \mathbbm{E}\Big[G(\tau^{\ubar{\lambda}^{\ast}}_{n,t},V_{\ubar{\lambda}^{\ast}}^{v_0, \pi}(\tau^{\ubar{\lambda}^{\ast}}_{n,t}),z(\tau^{\ubar{\lambda}^{\ast}}_{n,t}))\ \Big | \ \mathcal{F}_t \Big] \\
    &\hspace{1.5 mm} \leq  \lim_{n\rightarrow \infty}G\big(t,V_{\ubar{\lambda}^{\ast}}^{v_0, \pi}(t), z(t)\big) \\
   &\hspace{1.5 mm} = G\big(t,V_{\ubar{\lambda}^{\ast}}^{v_0, \pi}(t), z(t)\big).
\end{align*}
Conditioning on  $V_{\ubar{\lambda}^{\ast}}^{v_0, \pi}(t)=v$, $z(t)=z$ leads to (\ref{eq: dual optimality la}).\\
Letting $(\lambda, \ubar{\pi}^{\ast})\in \mathcal{D}\times \ubar{\Lambda}$ satisfy Condition \ref{cond: uniform integrability condition for dual} and following the analogous steps as before, we can prove equation (\ref{eq: dual optimality pi ast}). \\
As all previous inequalities become equalities if we instead consider $(\ubar{\lambda}^{\ast}, \ubar{\pi}^{\ast})\in \ubar{\mathcal{D}}\times \ubar{\Lambda}$ satisfying Condition \ref{cond: uniform integrability condition for dual}, equation (\ref{eq: dual equality la and pi}) follows immediately.
\end{proof}
\vspace{0.5cm}
\begin{proof}[Proof of Theorem \ref{thm: solution to dual HJB PDE given condition EAS}]
	We first determine the derivatives of $G=G(t,v,z)$ in terms of $A=A(T-t)$ and $B=B(T-t)$ as
	\begin{equation}
		\begin{aligned}\label{eq: derivatives exponential affine separability}
			G_t &= -\left(A_{\tau} + B_{\tau}'z\right)G, \quad G_v = \frac{b}{v}G, \quad G_{vv} = \frac{b(b-1)}{v^2}G \\
			\nabla_zG &= GB, \quad \nabla^2_z G = GBB', \quad \nabla_z\left(G_v\right)= \frac{b}{v}G B.
		\end{aligned}
	\end{equation}
	Plugging the derivatives into the dual HJBI PDE (\ref{eq: dual HJB PDE general}), factoring $\frac{bG}{2v(1-b)}>0$ out of the minimization, dividing by $G\neq 0$  and plugging in the optimizer $\hat{\lambda}^{\ast} = \hat{\lambda}^{\ast}(t,z,B)$ yields 
	\begin{align*}
		(\ast) :&=G_t + v r G_v + \left(\mu^z\right)'\left(\nabla_z G\right) + \frac{1}{2}\text{Trace}\left[\Sigma^z\left( \Sigma^z\right) '\nabla^2_zG\right] \nonumber \\
		& \quad + v \inf_{\lambda \in \R^d}\Big \{ \delta_K (\lambda) G_v- \frac{1}{2}\frac{1}{vG_{vv}}\Vert \Sigma^{-1}\left[\mu + \lambda - r\1\right]G_v + \left(\Sigma^z\rho\right)'\nabla_z\left(G_v\right)\Vert^2 \Big \} \\
		&= -\left(A_{\tau} - B_{\tau}'z\right)G + brG + (\mu^z)'BG+ \frac{1}{2}\text{Trace}\left[\Sigma^z\left( \Sigma^z\right) 'BB'\right]G  \\
		& \quad + v \inf_{\lambda \in \R^d}\Big \{ \delta_K (\lambda) \frac{b}{v}G- \frac{1}{2}\frac{v}{b(b-1)G} \left \Vert \Sigma^{-1}\left[\mu + \lambda - r\1\right]\frac{b}{v}G + \left(\Sigma^z\rho\right)'\frac{b}{v}BG \right \Vert^2 \Big \} \\
		&= -\left(A_{\tau} - B_{\tau}'z\right)G + brG + (\mu^z)'BG+ \frac{1}{2}\text{Trace}\left[\Sigma^z\left( \Sigma^z\right) 'BB'\right]G  \\
		& \quad + \frac{1}{2}\frac{bG}{1-b} \inf_{\lambda \in \R^d}\Big \{2(1-b) \delta_K (\lambda) +  \left \Vert \Sigma^{-1}\left[\mu + \lambda - r\1\right] + \left(\Sigma^z\rho\right)'B \right \Vert^2 \Big \} \\
		&= -A_{\tau} - B_{\tau}'z + br + (\mu^z)'B+ \frac{1}{2}\text{Trace}\left[\Sigma^z\left( \Sigma^z\right) 'BB'\right]  \\
		& \quad + \frac{1}{2}\frac{b}{1-b} \inf_{\lambda \in \R^d}\Big \{ 2(1-b)\delta_K (\lambda) +  \left \Vert \Sigma^{-1}\left[\mu + \lambda - r\1\right] + \left(\Sigma^z\rho\right)'B \right \Vert^2 \Big \}\\
		&= -A_{\tau} - B_{\tau}'z + br + (\mu^z)'B+ \frac{1}{2}\text{Trace}\left[\Sigma^z\left( \Sigma^z\right) 'BB'\right]  \\
		& \quad + \frac{1}{2}\frac{b}{1-b} \inf_{\lambda \in \R^d}\Big \{ 2(1-b)\delta_K (\lambda) + 2\lambda'\left(\Sigma\Sigma'\right)^{-1}\left(\mu  - r\1   + \left(\Sigma^z\rho\Sigma'\right)'B\right) + \left \Vert \Sigma^{-1} \lambda \right \Vert^2 \Big \}\\
		& \quad + \frac{1}{2}\frac{b}{1-b}\left \Vert \Sigma^{-1}\left[\mu  - r\1\right] + \left(\Sigma^z\rho\right)'B \right \Vert^2 \\
		&= -A_{\tau} - B_{\tau}'z + br + (\mu^z)'B+ \frac{1}{2}\text{Trace}\left[\Sigma^z\left( \Sigma^z\right) 'BB'\right]  \\
		& \quad + \frac{1}{2}\frac{b}{1-b} \left( 2(1-b)\delta_K\big(\hat{\lambda}^{\ast}\big) + 2\big(\hat{\lambda}^{\ast}\big)'\left(\Sigma\Sigma'\right)^{-1}\left(\mu  - r\1   + \left(\Sigma^z\rho\Sigma'\right)'B\right) + \left \Vert \Sigma^{-1} \hat{\lambda}^{\ast} \right \Vert^2 \right)\\
		& \quad + \frac{1}{2}\frac{b}{1-b}\left \Vert \Sigma^{-1}\left[\mu  - r\1\right] + \left(\Sigma^z\rho\right)'B \right \Vert^2 \\
		&= -A_{\tau} - B_{\tau}'z + b\left(r+ \delta_K\big(\hat{\lambda}^{\ast}\big)\right)+ (\mu^z)'B+ \frac{1}{2}\text{Trace}\left[\Sigma^z\left( \Sigma^z\right) 'BB'\right]  \\
		&\quad + \frac{1}{2}\frac{b}{1-b} \left( \left \Vert \Sigma^{-1}\big(\mu + \hat{\lambda}^{\ast}-r\1\big) \right \Vert^2 + 2B'\Sigma^z\rho \Sigma^{-1}\left(\mu+\hat{\lambda}^{\ast}-r\1\right)+B'\Sigma^z \rho\left(\Sigma^z \rho\right)'B\right)
	\end{align*}
	Condition \ref{cond: Exp Separability optimal lambda*} allows us to replace the market coefficients by affine functions in $z$. Thus, we obtain
	\begin{align*}
		(\ast)&= -A_{\tau} - B_{\tau}'z + b\left(p_0 + p_1'z\right)+ \left(k_0 + k_1z\right)'B+ \frac{1}{2}\text{Trace}\left[\left(h_0 + h_1[z]\right)BB'\right]  \\
		&\quad + \frac{1}{2}\frac{b}{1-b} \Big( q_0 + q_1'z + 2B'\left(g_0 + g_1z\right)+  B'\left(l_0+h_0 +l_1[z]+h_1[z]\right)B\Big) \\
		&= - A_{\tau} +bp_0 +k_0'B + \frac{1}{2}\text{Trace}\left[h_0BB'\right] + \frac{1}{2}\frac{b}{1-b} \Big( q_0 +2g_0'B+B'\left(l_0+h_0\right)B \Big)\\
		&\quad -B_{\tau}'z +bp_1'z+B'k_1z+\frac{1}{2}\text{Trace}\left[h_1[z]BB'\right] +\frac{1}{2}\frac{b}{1-b} \Big(q_1'z+2B'g_1z+B'\left(l_1[z]+h_1[z]\right)B\Big)
	\end{align*}
	Making use of the commutive property of the trace of a matrix as well as the matrix representation of $h_1[\cdot]$ and  $l_1[\cdot]$, we have
	\begin{align*}
		\text{Trace}\left[h_0BB'\right] &= \text{Trace}\big[\underbrace{B'h_0B}_{\in \R}\big] = B'h_0B\\
		\text{Trace}\left[h_1[z]BB'\right] &= \text{Trace}\big[\underbrace{B'h_1[z]B}_{\in \R}\big] = B'h_1[z]B = \left(B'h_1[\cdot]B\right)'z\\
		B'\left(l_1[z]+h_1[z]\right)B&= \left(B'\left(l_1[\cdot]+h_1[\cdot]\right)B\right)'z.
	\end{align*}
	Thus, as $A$ and $B$ are solutions to the ODEs (\ref{eq: constr ODE for A}) and (\ref{eq: constr ODE for B})
	\begin{align*}
		(\ast)&=- A_{\tau} +bp_0 +k_0'B + \frac{1}{2}B'h_0B+ \frac{1}{2}\frac{b}{1-b} \Big( q_0 +2g_0'B+B'\left(l_0+h_0\right)B \Big) \\
		&\quad + \left(-B_{\tau} + bp_1 + k_1'B + \frac{1}{2}B'h_1[\cdot]B + \frac{1}{2}\frac{b}{1-b} \Big(q_1 +2g_1'B+B'\left(l_1[\cdot]+h_1[\cdot]\right)B\Big)\right)'z \\
		&= 0.
	\end{align*}
	Hence, $G$ is a solution to the dual HJBI PDE (\ref{eq: dual HJB PDE general}) and thereby, according to Lemma \ref{lem: dual HJB PDE}, also a solution to the primal HJB PDE (\ref{eq: constr. HJB PDE Factor Model}).
\end{proof}
\vspace{0.5cm}
\begin{proof}[Proof of Theorem \ref{thm: verification theorem primal OP}]
	We first derive the explicit expression (\ref{eq: optimal constr. pi (general)}) for $\ubar{\pi}^{\ast}$ in terms of $\ubar{\lambda}^{\ast}$ and $B$. Following the arguments in the proof of Lemma \ref{lem: dual HJB PDE}, we realize that the candidate optimal portfolio $\ubar{\pi}^{\ast}$ is given as the minimizing argument $\pi_{\la},$ i.e.
	\begin{align}
		\ubar{\pi}^{\ast} = \pi_{\la} =  -\frac{1}{vG_{vv}}\left(\Sigma\Sigma'\right)^{-1}\left[G_v(\mu + \ubar{\lambda}^{\ast} - r \1) + \Sigma \rho'\left(\Sigma^z\right)'\nabla_z\left(G_v\right) \right]. \label{eq: optimal pi before plugging in seperability}
	\end{align}
	On the other hand, since Condition \ref{cond: Exp Separability optimal lambda*} is satisfied, the exponentially affine structure of $G$ from (\ref{eq: exp affine separable value function}) implies that the derivatives of $G$ are given by (\ref{eq: derivatives exponential affine separability}). Plugging these derivatives into (\ref{eq: optimal pi before plugging in seperability}) yields the candidate optimal portfolio
	\begin{align*}
		\ubar{\pi}^{\ast} &= -\left(\frac{b(b-1)}{v}G\right)^{-1}\left(\Sigma\Sigma'\right)^{-1}\left[\frac{b}{v}G(\mu + \ubar{\lambda}^{\ast} - r \1) + \Sigma \rho'\left(\Sigma^z\right)' \frac{b}{v}G B  \right] \\
		&= \frac{1}{1-b}\left(\Sigma\Sigma'\right)^{-1}\left[\mu + \ubar{\lambda}^{\ast} - r \1 + \left(\Sigma^z\rho\Sigma'\right)'B\right].
	\end{align*}
	We continue by proving  (\ref{eq: G attains value function for optimal strategy}). As $G$ is a solution to the (dual) HJB PDE (\ref{eq: dual HJB PDE general}) and  $\ubar{\pi}^{\ast},$ $\ubar{\lambda}^{\ast}$ attain the optimum in (\ref{eq: dual HJB PDE general}), we know for every $t\in [0,T]$
	\begin{align}\label{eq: SDE wealth process in value function}
		dG(t,V^{v_0,\ubar{\pi}^{\ast}}(t),z(t)) & = V^{v_0,\ubar{\pi}^{\ast}}(t)G_v\left(t,V^{v_0,\ubar{\pi}^{\ast}}(t),z(t)\right) \left(\ubar{\pi}^{\ast}(t,V^{v_0,\ubar{\pi}^{\ast}}(t),z(t))\right)'\Sigma dW(t) \nonumber \\
		&  \qquad + \left(\nabla_z\left(G\right)(t,V^{v_0,\ubar{\pi}^{\ast}}(t),z(t))\right)'\Sigma^zdW^z(t). 
	\end{align}
	In particular, noting that $V^{v_0,\ubar{\pi}^{\ast}} = V_0^{v_0,\ubar{\pi}^{\ast}}$ and considering the stopping times $\tau^{0}_{n,t}$ for $n\in \N$ from Condition $(\text{UI}_0)$ yields
	\begin{align*}
	 \mathbbm{E}\Big[U(V^{v_0,\ubar{\pi}^{\ast}}(T))\ \Big | \ \mathcal{F}_t \Big] & \hspace{1.7 mm} =  \mathbbm{E}\Big[G\left(T,V^{v_0,\ubar{\pi}^{\ast}}(T),z(T)\right)\ \Big |  \mathcal{F}_t \Big] \\
		& \hspace{1.7 mm} = \mathbbm{E}\Big[\lim_{n\rightarrow \infty}G\left(\tau^{0}_{n,t},V^{v_0,\ubar{\pi}^{\ast}}(\tau^{0}_{n,t}),z(\tau^{0}_{n,t})\right) \ \Big |  \mathcal{F}_t \Big] \\
		&\overset{(\text{UI}_0)}{=} 
		\lim_{n\rightarrow \infty} \mathbbm{E}\Big[G\left(\tau^{0}_{n,t},V^{v_0,\ubar{\pi}^{\ast}}(\tau^{0}_{n,t}),z(\tau^{0}_{n,t})\right) \ \Big |  \mathcal{F}_t \Big] \\
		&\hspace{1.7 mm} = \lim_{n\rightarrow \infty} \Big(G\left(t,V^{v_0,\ubar{\pi}^{\ast}}(t),z(t)\right) \\
		& \quad \hspace{1.7 mm} + \underbrace{\mathbbm{E}\Big[\int_t^{\tau^{0}_{n,t}} V^{v_0,\ubar{\pi}^{\ast}}(u)G_v\left(u,V^{v_0,\ubar{\pi}^{\ast}}(u),z(u)\right) \left(\ubar{\pi}^{\ast}(u,V^{v_0,\ubar{\pi}^{\ast}}(u),z(u))\right)'\Sigma(u,z(u)) dW(u)\Big |  \mathcal{F}_t \Big]}_{= 0, \ \text{by choice of }\tau^{0}_{n,t}}\\
		& \quad \hspace{1.7 mm}+ \underbrace{\mathbbm{E}\Big[\int_t^{\tau^{0}_{n,t}}\left(\nabla_z\left(G\right)(u,V^{v_0,\ubar{\pi}^{\ast}}(u),z(u))\right)'\Sigma^z(u,z(u)) dW^z(u)\Big |  \mathcal{F}_t \Big]}_{= 0, \ \text{by choice of }\tau^{0}_{n,t}} \Big) \\
		&\hspace{1.7 mm}= G\left(t,V^{v_0,\ubar{\pi}^{\ast}}(t),z(t)\right).
	\end{align*}
	Conditioning on $V^{v_0,\ubar{\pi}^{\ast}}(t)=v$, $z(t)=z$ leads to equation (\ref{eq: G attains value function for optimal strategy}). \\
	
	It remains to show that $\ubar{\pi}^{\ast}$ dominates all other $\pi \in \Lambda_K$, i.e. we need to verify inequality (\ref{eq: pi star dominates all other constrained pi}). The proof idea is similar to that of Theorem 4.3 in \cite{Escobar2021}, but adapted for our constrained setting. Let $t\in [0,T]$ and $\pi \in \Lambda_K(t)$ be arbitrary but fixed. Define the process $L = \left(L(u)\right)_{u\in [t,T]}$ as
	\begin{align*}
		L(u) = b \frac{V^{v_0,\pi}(u)}{V^{v_0,\ubar{\pi}^{\ast}}(u)}G(u, V^{v_0,\ubar{\pi}^{\ast}}(u), z(u)).
	\end{align*} 
	We proceed by first deriving the SDE of $L$ and then showing that $L$ is a supermartingale. By It\^{o}'s product rule we have with
	$\pi^{\ast}(u) = \ubar{\pi}^{\ast}(u,V^{v_0,\ubar{\pi}^{\ast}}(u),z(u))$
	\begin{align*}
		d&\left(\frac{V^{v_0,\pi}(u)}{V^{v_0,\ubar{\pi}^{\ast}}(u)}\right) = \frac{1}{V^{v_0,\ubar{\pi}^{\ast}}(u)}dV^{v_0,\pi}(u) + V^{v_0,\pi}(u) d\left(\frac{1}{V^{v_0,\ubar{\pi}^{\ast}}(u)}\right) + d\langle V^{v_0,\pi}(u) ,\frac{1}{V^{v_0,\ubar{\pi}^{\ast}}(u)}\rangle_u \\
		&= \left(\frac{V^{v_0,\pi}(u)}{V^{v_0,\ubar{\pi}^{\ast}}(u)}\right) \Bigg(\left[r + (\mu-r\1)'\pi(u)\right]du + \pi(u)'\Sigma dW(u) \\
		&\hspace{30mm}-\left[r + (\mu-r\1)'\pi^{\ast}(u)-\Vert \Sigma'\pi^{\ast}(u)\Vert^2  \right]du - \pi^{\ast}(u)'\Sigma dW(u) - \pi^{\ast}(u)'\Sigma\Sigma' \pi(u)du	\Bigg) \\
		&= \left(\frac{V^{v_0,\pi}(u)}{V^{v_0,\ubar{\pi}^{\ast}}(u)}\right) \Bigg(\Big[(\mu-r\1)'(\pi(u)-\pi^{\ast}(u)) \underbrace{- \pi^{\ast}(u)'\Sigma\Sigma' \pi(u) + \Vert \Sigma'\pi^{\ast}(u)\Vert^2}_{= - \pi^{\ast}(u)'\Sigma\Sigma'(\pi(u)-\pi^{\ast}(u))} \Big]du + (\pi(u)-\pi^{\ast}(u))'\Sigma dW(u)   \Bigg) \\
		&= \left(\frac{V^{v_0,\pi}(u)}{V^{v_0,\ubar{\pi}^{\ast}}(u)}\right) (\pi(u)-\pi^{\ast}(u))' \Bigg(\left[\mu-r\1 - \Sigma \Sigma'\pi^{\ast}(u)\right]du + \Sigma dW(u)\Bigg).
	\end{align*}
	Moreover, due to (\ref{eq: derivatives exponential affine separability}) and (\ref{eq: SDE wealth process in value function}),
	\begin{align*}
		dG(u,V^{v_0,\ubar{\pi}^{\ast}}(u),z(u)) \overset{(\ref{eq: SDE wealth process in value function})}{=} G(u,V^{v_0,\ubar{\pi}^{\ast}}(u),z(u))\Big( b\pi^{\ast}(u)'\Sigma dW(u) + B(T-u)'\Sigma^zdW^z(u)\Big)
	\end{align*}
	and 
	\begin{align*}
	 d\langle \frac{V^{v_0,\pi}}{V^{v_0,\ubar{\pi}^{\ast}}} , G(\cdot,V^{v_0,\ubar{\pi}^{\ast}},z)\rangle_u  = \frac{V^{v_0,\pi}(u)}{V^{v_0,\ubar{\pi}^{\ast}}(u)}G(u,V^{v_0,\ubar{\pi}^{\ast}}(u),z(u))(\pi(u)-\pi^{\ast}(u))'\Sigma \left(b\Sigma'\pi^{\ast}(u) + \left(\Sigma^z \rho\right)'B(T-u)\right)du.
	\end{align*}
	In total, this yields with $\la(u) = \ubar{\lambda}^{\ast}(u,V^{v_0,\ubar{\pi}^{\ast}}(u), z(u))$
	\begin{align*}
		dL(u) &= bd\left(\frac{V^{v_0,\pi}(u)}{V^{v_0,\ubar{\pi}^{\ast}}(u)}G(u,V^{v_0,\ubar{\pi}^{\ast}}(u),z(u))\right) \\
		&= b\Big[\frac{V^{v_0,\pi}(u)}{V^{v_0,\ubar{\pi}^{\ast}}(u)}dG(u,V^{v_0,\ubar{\pi}^{\ast}}(u),z(u))   G(u,V^{v_0,\ubar{\pi}^{\ast}}(u),z(u)) d\left(\frac{V^{v_0,\pi}(u)}{V^{v_0,\ubar{\pi}^{\ast}}(u)}\right) + d\langle \frac{V^{v_0,\pi}}{V^{v_0,\ubar{\pi}^{\ast}}} , G(\cdot,V^{v_0,\ubar{\pi}^{\ast}},z)\rangle_u\Big] \\
		&= L(u) \Big[ (\pi(u)-\pi^{\ast}(u))'\left[\mu-r\1 - \Sigma \Sigma'\pi^{\ast}(u)\right]du + \Sigma dW(u)  +  b\pi^{\ast}(u)'\Sigma dW(u) + B(T-u)'\Sigma^zdW^z(u) \\
		&\qquad  \qquad + (\pi(u)-\pi^{\ast}(u))'\Sigma \left(b\Sigma'\pi^{\ast}(u) + \left(\Sigma^z \rho\right)'B(T-u)\right)du \Big] \\
		&= L(u)(\pi(u)-\pi^{\ast}(u))'\underbrace{\Big[\mu - r\1 + \left(\Sigma^z\rho\Sigma'\right)'B(T-u)- (1-b)\Sigma\Sigma'\pi^{\ast}(u)\Big]}_{\overset{(\ref{eq: optimal constr. pi (general)})}{=}-\la(u)}du \\
		& \quad + L(u)\Big((\pi(u)-(1-b)\pi^{\ast}(u))'\Sigma dW(u) + B(T-u)'\Sigma^zdW^z(u)\Big). \\
		&= L(u) \Big( (\pi^{\ast}(u)-\pi(u))'\la(u)du + (\pi(u)-(1-b)\pi^{\ast}(u))'\Sigma dW(u) + B(T-u)'\Sigma^zdW^z(u)\Big)
	\end{align*}
	However, Remark \ref{rem: slackness condition} implies
	\begin{align*}
		&\delta_K(\ubar{\lambda}^{\ast}(u,v,z))+\ubar{\pi}^{\ast}(u,v,z)'\ubar{\lambda}^{\ast}(u,v,z) = 0 \quad \forall (u,v,z) \in [0,T]\times (0,\infty)\times \R^m \\
		\Rightarrow \ & \delta_K(\la(u)) + \la(u)'\pi^{\ast}(u) = 0 \quad \mathcal{L}[t,T]\otimes Q-\text{a.e.}. 
	\end{align*}
	Moreover, as $\pi\in \Lambda_K(t)$, $\pi(u)\in K$ holds $\mathcal{L}[t,T]\otimes Q-\text{a.e.}$ and thus
	\begin{align*}
		\delta_K(\la(u)) + \la(u)'\pi(u) = -\underbrace{\inf_{x \in K}\left(\la(u)'x\right)}_{\leq \la(u)'\pi(u)}+ \la(u)'\pi(u) \geq 0 \quad \mathcal{L}[t,T]\otimes Q-\text{a.e.}.
	\end{align*}
	We finally obtain
	\begin{align*}
		dL(u) &= L(u)\Big(\underbrace{ -\left(\delta_K(\la(u)) + \la(u)'\pi(u)\right)}_{= (\ast) }du + (\pi(u)-(1-b)\pi^{\ast}(u))'\Sigma dW(u) + B(T-u)'\Sigma^zdW^z(u)\Big)
	\end{align*}
	and for any $s \in [t,T]$ , $L(s)$ can therefore be expressed as
	\begin{align*}
		L(s) = L(t)\exp\left(-\int_t^s \delta_K(\la(u)) + \la(u)'\pi(u) du\right)M(s),
	\end{align*}
	for a supermartingale $M = \left(M(u)\right)_{u\in [t,T]}$ which satisfies the SDE
	$$
	dM(u) = M(u)\left((\pi(u)-(1-b)\pi^{\ast}(u))'\Sigma dW(u) + B(T-u)'\Sigma^zdW^z(u)\right), \quad M(t)=1.
	$$
	Hence, as $(\ast)\leq 0$, $L$ is a supermartingale, too. \\
	
	To conclude the proof, recall that $U$ is concave and therefore $U(y)\leq U(x) + U'(x)(y-x)$ for all $x,y \in (0,\infty)$. This leads to 
	\begin{align*}
		\E\left[U(V^{v_0, \pi}(T))\  \big | \ \mathcal{F}_t  \right]  & \leq \E \left[U(V^{v_0, \ubar{\pi}^{\ast}}(T))\  \big | \ \mathcal{F}_t  \right] +  \E \left[U'(V^{v_0, \ubar{\pi}^{\ast}}(T))\left(V^{v_0, \pi}(T) - V^{v_0, \ubar{\pi}^{\ast}}(T) \right) \ \big | \ \mathcal{F}_t  \right] \\
		&= \E \left[U(V^{v_0, \ubar{\pi}^{\ast}}(T))\  \big | \ \mathcal{F}_t  \right] +  \underbrace{\E \left[L(T) \  \big | \ \mathcal{F}_t  \right]}_{\leq L(t)} -  \underbrace{\E \left[ bU(V^{v_0, \ubar{\pi}^{\ast}}(T)) \ \big | \ \mathcal{F}_t \right]}_{= bG(t,V^{v_0, \ubar{\pi}^{\ast}}(t),z(t))} \\
		& \leq \E \left[U(V^{v_0, \ubar{\pi}^{\ast}}(T))\  \big | \ \mathcal{F}_t  \right] + L(t)-bG(t,V^{v_0, \ubar{\pi}^{\ast}}(t),z(t)) \\
		&= \E \left[U(V^{v_0, \ubar{\pi}^{\ast}}(T))\  \big | \ \mathcal{F}_t  \right] + b\left(G(t,V^{v_0, \pi}(t),z(t))\frac{V^{v_0, \pi}(t)}{V^{v_0, \ubar{\pi}^{\ast}}(t)}-G(t,V^{v_0, \ubar{\pi}^{\ast}}(t),z(t))\right).
	\end{align*}
	Finally, conditioning on $V^{v_0, \pi}(t) = V^{v_0, \ubar{\pi}^{\ast}}(t) = v$, $z(t)=z$ yields (\ref{eq: pi star dominates all other constrained pi}).
\end{proof}
\vspace{0.5cm}
\begin{proof}[Proof of Lemma \ref{lem: solution to dual HJBI PDE B.S. model}]
	In $\mathcal{M}_{BS}$, the optimizer $\hat{\lambda}^{\ast}$ from (\ref{eq: minimization w.r.t. lambda}) is given as
	\begin{align*}
		\hat{\lambda}^{\ast}(t,z,B) &= \underset{\lambda \in \R^d}{\text{argmin}}\Bigg \{ 2(1-b)\delta_K(\lambda) + \Big \Vert \underbrace{\Sigma(t,z)^{-1}}_{\sigma(t)^{-1}}\big(\underbrace{\mu(t,z)-r(t,z)\1}_{=\eta(t)} + \lambda \big) + \underbrace{\left(\Sigma^z(t,z)\rho(t,z)\right)'}_{= 0}B \Big\Vert^2 \Bigg \} \\
		&= \underset{\lambda \in \R^d}{\text{argmin}}\left \{ 2(1-b)\delta_K(\lambda) + \left \Vert \sigma(t)^{-1}\left( \eta(t) + \lambda \right) \right \Vert^2\right \} \\
		&= \la(t),
	\end{align*}
	i.e. $\hat{\lambda}^{\ast}(t,v,z) = \la(t)$ is a deterministic function independent of $z$ and $B$. Moreover, by defining
	\begin{align*}
		p_0(t,B) &= r(t) + \delta_K\left(\la(t)\right) \\
		q_0(t,B) &= \left \Vert \sigma(t)^{-1}\left(\eta(t)+\la(t)\right) \right \Vert^2 
	\end{align*}
	and setting the remaining coefficients $k_0,$  $k_1,$ $h_0,$ $h_1,$ $l_0,$ $l_1,$ $p_1,$ $q_1,$ $g_0$ and $g_1$ to zero, Condition \ref{cond: Exp Separability optimal lambda*} is satisfied. The corresponding ODEs (\ref{eq: constr ODE for A}) and (\ref{eq: constr ODE for B}) simplify to
	\begin{align*}
		A_{\tau}(\tau) &= b\left(r(T-\tau) + \delta_K\left(\la(T-\tau)\right)\right)+\frac{1}{2}\frac{b}{1-b} \left \Vert \sigma(t)^{-1}\left( \eta(T-\tau) + \la(T-\tau)\right) \right \Vert^2  \\
		&= br(T-\tau)+\frac{1}{2}\frac{b}{1-b} \inf_{\lambda \in \R^d}\Big\{ 2(1-b)\delta_K(\lambda) + \left \Vert \sigma(T-\tau)^{-1}\left(\eta(T-\tau)+\lambda \right)\right \Vert^2 \Big \}\\
		B_{\tau}(\tau) &=0.
	\end{align*}
	Hence, $B \equiv 0$ and $A$ can be obtained through simple integration. In particular, the candidate optimal portfolio (\ref{eq: optimal constr. pi (general)}) is given through
	\begin{align*}
		\ubar{\pi}^{\ast}(t,v,z)= \frac{1}{1-b}\left(\sigma(t)\sigma(t)'\right)^{-1}\left(\eta(t) + \la(t)\right).
	\end{align*}
\end{proof}
\vspace{0.5cm}
\begin{proof}[Proof of Corollary \ref{cor: uniform integrability B.S. model}]
	We verify Condition \ref{cond: uniform integrability condition for dual} by showing the $L^q$ boundedness of $G\left(\tau^0_{n,t}, V^{v_0, \ubar{\pi}^{\ast}}_{0}(\tau^0_{n,t}),z(\tau^0_{n,t})\right)$ in $n\in \N$ for arbitrary $q>1.$ \\
	The market coefficients $\sigma(t),$ $\eta(t)$ and $r(t)$ are continuous and therefore uniformly bounded in $t\in [0,T].$ In particular, this has the consequence that $\la(t)$ is uniformly bounded in $t\in[0,T],$ since the quadratic term in (\ref{eq: minimization w.r.t. lambda}) dominates for large $\Vert \lambda \Vert.$ Hence, $A(T-t)$ and the candidate optimal portfolio $\pi^{\ast}(t):= \ubar{\pi}^{\ast}(t,v,z)$ are uniformly bounded in $t\in [0,T],$ too. 
	For arbitrary $q>1,$ we can thus find a constant $C_q > 0$ such that for all $t\in [0,T]$
	\begin{align}
		\Big | G&(t,V^{v_0, \ubar{\pi}^{\ast}}(t),z(t)) \Big |^q \overset{\text{Th. \ref{thm: solution to dual HJB PDE given condition EAS}}}{=}\left | \frac{1}{b}\exp\left(b \ln\left(V^{v_0,\ubar{\pi}^{\ast}}(t)\right)+A(T-t)+B(T-t)'z(t)\right)\right |^q\nonumber \\
		&= \frac{1}{|b|}\exp \Bigg(bq \int_0^t r(s) + \eta(s)'\pi^{\ast}(s)-\frac{1}{2}\Vert \sigma(s)'\pi^{\ast}(s)\Vert^2ds  + bq \int_0^t \pi^{\ast}(s)'\sigma(s)dW(s) + qA(T-t)\Bigg) \nonumber \\
		&=  \frac{1}{|b|}\exp \Bigg(bq \int_0^t r(s) + \eta(s)'\pi^{\ast}(s)-\frac{1-bq}{2}\Vert \sigma(s)'\pi^{\ast}(s)\Vert^2ds + q A(T-t)\nonumber \\
		& \hspace{2.5cm}- \frac{1}{2}\int_0^t b^2q^2\Vert \sigma(s)'\pi^{\ast}(s)\Vert^2ds + bq \int_0^t \pi^{\ast}(s)'\sigma(s)dW(s) \Bigg) \nonumber \\
		&\leq C_q \underbrace{\exp \left(- \frac{1}{2}\int_0^t b^2q^2\Vert \sigma(s)'\pi^{\ast}(s)\Vert^2ds + bq \int_0^t \pi^{\ast}(s)'\sigma(s)dW(s) \right)}_{=: M_t} \nonumber\\
		&= C_q M_t. \label{eq: bound by exp. supermartingale BS}
	\end{align}
	The process $M = \left(M_t\right)_{t\in [0,T]}$ is a non-negative local martingale and thus a supermartingale. Doob's optional sampling theorem (\enquote{O.S.}) implies
	\begin{align*}
		\sup_{n\in \N} \E \left[ \left |G\left(\tau^0_{n,t}, V^{v_0, \ubar{\pi}^{\ast}}_{0}(\tau^0_{n,t}),z(\tau^0_{n,t})\right) \right |^q\right] \overset{(\ref{eq: bound by exp. supermartingale BS})}{\leq}C_q \sup_{n\in \N} \E \left[ M_{\tau^0_{n,t}} \right] \overset{O.S.}{\leq} C_q M_0 = C_q < \infty.
	\end{align*}
	Hence, 
	$$
	\left(G\left(\tau^0_{n,t}, V^{v_0, \ubar{\pi}^{\ast}}_{0}(\tau^0_{n,t}),z(\tau^0_{n,t})\right)\right)_{n\in \N}
	$$
	is bounded in $L^q$ for any $q>1$ and $t\in [0,T]$ and is thus uniformly integrable for any $t\in [0,T]$ (see Theorem 4.5.9 in \cite{Bogachev2007} with $G(t) = t^q$). Hence, Condition \ref{cond: uniform integrability condition for dual} is satisfied and $\ubar{\pi}^{\ast}$ is optimal for $\mathbf{(P)}$ by virtue of Theorem \ref{thm: verification theorem primal OP}.
\end{proof}
\vspace{0.5cm}
\begin{proof}[Proof of Lemma \ref{lem: solution to dual HJBI PDE CIR model}]
	We may rewrite any $\lambda \in \R^d$ as 
	\begin{align*}
		\lambda = \begin{pmatrix}
			\lambda_1 \\
			\vdots \\
			\lambda_m
		\end{pmatrix},
		\quad \text{with} \quad \lambda_i \in \R^{d_i}, \quad i=1,...,m.
	\end{align*}
	Hence, we obtain for every $\lambda \in \R^d$
	\begin{align*}
		\delta_K(\lambda) = -\inf_{x\in K}\left(x'\lambda\right)= -\inf_{\substack{x_i \in K_i \\ i=1,...,m}}\left(\sum_{i=1}^m x_i'\lambda_i \right) = -\sum_{i=1}^m \inf_{x_i \in K_i} \left(x_i'\lambda_i\right) = \sum_{i=1}^m \delta_{K_i}(\lambda_i).
	\end{align*}
	Note that we may restrict the minimization (\ref{eq: minimization w.r.t. lambda}) to $z\in (0,\infty)^m$ because Feller's condition (\ref{eq: Fellers Condition in m-factor CIR market}) ensures that the $m$-dimensonal CIR process has positive components $\mathcal{L}[0,T]\otimes Q$-a.e.. Hence, for any $(t,z,B)\in [0,T]\times (0,\infty)^m\times \R^m$ we have
	\begin{align*}
		\left[\left(\Sigma(t,z)\Sigma(t,z)'\right)^{-1}\lambda\right]'\left[\mu(t,z)-r(t,z)\1\right] 	&= \begin{pmatrix}
			\frac{1}{z_1}\left(\Sigma_1\Sigma_1'\right)^{-1}\lambda_1 \\
			\vdots \\
			\frac{1}{z_m}\left(\Sigma_m\Sigma_m'\right)^{-1}\lambda_m 
		\end{pmatrix}'\begin{pmatrix}
			\eta_1 z_1 \\
			\vdots \\
			\eta_m z_m
		\end{pmatrix} \\
		&  = \sum_{i=1}^m \lambda_i \left(\Sigma_i\Sigma_i'\right)^{-1}\eta_i,
	\end{align*}
	\begin{align*}
		\left[\Sigma(t,z)^{-1}\lambda\right]'\rho'\Sigma^z(t,z)B &= \begin{pmatrix}
			\frac{1}{\sqrt{z_1}}\Sigma_1^{-1}\lambda_1 \\
			\vdots \\
			\frac{1}{\sqrt{z_m}}\Sigma_m^{-1}\lambda_m 
		\end{pmatrix}'\begin{pmatrix}
			\rho_1 & & 0 \\
			& \ddots & \\
			0 & & \rho_m
		\end{pmatrix}\begin{pmatrix}
			\sigma_1 \sqrt{z_1}B_1 \\
			\vdots \\
			\sigma_m \sqrt{z_m}B_m
		\end{pmatrix}  \\
		& =\begin{pmatrix}
			\frac{1}{\sqrt{z_1}}\Sigma_1^{-1}\lambda_1 \\
			\vdots \\
			\frac{1}{\sqrt{z_m}}\Sigma_m^{-1}\lambda_m 
		\end{pmatrix}'\begin{pmatrix}
			\sigma_1 \sqrt{z_1}B_1 \rho_1 \\
			\vdots \\
			\sigma_m \sqrt{z_1}B_m \rho_m
		\end{pmatrix} \\
		&= \sum_{i=1}^m \sigma_i B_i \left(\Sigma_i^{-1}\lambda_i\right)'\rho_i,
	\end{align*}
	and
	\begin{align*}
		\lambda'\left(\Sigma(t,z)\Sigma(t,z)'\right)^{-1}\lambda &= \sum_{i=1}^m \left \Vert \Sigma_i^{-1}\lambda_i \right \Vert^2 z_i.
	\end{align*}
	Hence, the minimizer of (\ref{eq: minimization w.r.t. lambda}) can be rewritten as
	\begin{align}
		\hat{\lambda}^{\ast}(t,z,B) &=  \underset{\lambda \in \R^d}{\text{argmin}}\Big \{ 2(1-b)\delta_K(\lambda) + 2\lambda'\left(\Sigma(t,z)\Sigma(t,z)'\right)^{-1}\left[\mu(t,z) - r(t,z) \1\right] \nonumber \\
		& \hspace{2cm} +  2\lambda'\left(\Sigma(t,z)\Sigma(t,z)'\right)^{-1}(\Sigma^z(t,z)\rho(t,z)\Sigma(t,z)')'B + \left \Vert \Sigma(t,z)^{-1} \lambda \right \Vert^2\Big \} \nonumber \\
		&=  \underset{\substack{\lambda = \left(\lambda_1,...,\lambda_m\right)' \\	\lambda_i \in \R^{d_i}}}{\text{argmin}}\Big \{ \sum_{i=1}^m 2(1-b)\delta_{K_i}(\lambda_i) + 2\left(\Sigma_i^{-1}\lambda_i\right)'\left(\Sigma_i^{-1}\eta_i + \sigma_i B_i \rho_i \right) + \left \Vert \Sigma_i^{-1}\lambda_i \right \Vert^2 z_i \Big \} \nonumber \\
		&=  \underset{\substack{\lambda = \left(\lambda_1,...,\lambda_m\right)' \\		\lambda_i \in \R^{d_i}}}{\text{argmin}}\Bigg \{ 
		\sum_{i=1}^m  z_i \Bigg[ 2(1-b)\delta_{K_i}\left(\frac{\lambda_i}{z_i}\right) + 2\left(\Sigma_i^{-1}\frac{\lambda_i}{z_i}\right)'\left(\Sigma_i^{-1}\eta_i + \sigma_i B_i \rho_i \right) + \left \Vert \Sigma_i^{-1}\frac{\lambda_i}{z_i} \right \Vert^2 \Bigg]\Bigg \} \nonumber \\
		&= \underset{\substack{\lambda = \left(\lambda_1,...,\lambda_m\right)' \\		\lambda_i \in \R^{d_i}}}{\text{argmin}}\Bigg \{ \sum_{i=1}^m  z_i \Bigg[ 2(1-b)\delta_{K_i}\left(\frac{\lambda_i}{z_i}\right) + \left \Vert \Sigma_i^{-1}\left(\eta_i + \frac{\lambda_i}{z_i}\right) + \sigma_i B_i \rho_i\right \Vert^2 \Bigg] \Bigg \} \label{eq: optimizer in CIR market}
	\end{align}
	Using the change of control $\hat{\lambda}_i = \frac{\lambda_i}{z_i},$ we see that $\hat{\lambda}^{\ast}(t,z,B) = \la(t,z,B)$ from the statement of the Lemma.
	Letting $e_i$ denote the $i$-th unit vector in $\R^m,$ we can express the market coefficients in Condition \ref{cond: Exp Separability optimal lambda*} as
	\begin{align*}
		&\mu^{z}(t,z) = \kappa\odot \left(\theta-z\right) = \underbrace{\kappa\odot \theta}_{=:k_0(t)}+\underbrace{\begin{pmatrix}
				-\kappa_1 & & 0 \\
				& \ddots & \\
				0 & & - \kappa_m 
		\end{pmatrix}}_{=: k_1(t)}z \\
		&\Sigma^z(t,z)\Sigma^z(t,z)= \begin{pmatrix}
			\sigma_1^2z_1 & & 0 \\
			& \ddots & \\
			0 & & \sigma_m^2 z_m 
		\end{pmatrix} = \underbrace{\begin{pmatrix}
				z'(\sigma_1^2 e_1) & & 0 \\
				& \ddots & \\
				0 & & z'(\sigma_m^2 e_m) \end{pmatrix}}_{=: h_1(t)[z]} \\
		&\Sigma^z(t,z)\rho(t,z)\left(\Sigma^z(t,z) \rho(t,z)\right)'-	\Sigma^z(t,z)\Sigma^z(t,z)' \\
		& \qquad = \begin{pmatrix}
			\sigma_1 \sqrt{z_1}\rho_1' & & 0 \\
			& \ddots & \\
			0 & & \sigma_m \sqrt{z_m}\rho_m'
		\end{pmatrix}\begin{pmatrix}
			\sigma_1 \sqrt{z_1}\rho_1 & & 0 \\
			& \ddots & \\
			0 & & \sigma_m \sqrt{z_m}\rho_m 
		\end{pmatrix} - \begin{pmatrix}
			\sigma_1^2z_1 & & 0 \\
			& \ddots & \\
			0 & & \sigma_m^2 z_m 
		\end{pmatrix} \\
		& \qquad = \begin{pmatrix}
			\sigma_1^2z_1 \left(\Vert \rho_1 \Vert^2-1\right) & & 0 \\
			& \ddots & \\
			0 & & \sigma_m^2 z_m  \left(\Vert \rho_m \Vert^2-1\right)
		\end{pmatrix} \\
		& \qquad = \underbrace{\begin{pmatrix}
				z'\left(\sigma_1^2\left(\Vert \rho_1 \Vert^2-1\right)e_1\right) & & 0 \\
				& \ddots & \\
				0 & & z'\left(\sigma_m^2\left(\Vert \rho_m \Vert^2-1\right)e_m\right)
		\end{pmatrix}}_{=:l_1(t)[z]} \\
		&r(t,z) + \delta_K(\hat{\lambda}^{\ast}(t,z,B)) = \underbrace{r}_{=:p_0(t,B)} + \sum_{i=1}^m \underbrace{\delta_{K_i}\left(\la_i(B_i)\right)}_{=: \left(p_1(t,B)\right)_i}z_i  \\
		&\left\Vert\Sigma^{-1}(t,z)\left(\mu(t,z)+\hat{\lambda}^{\ast}(t,z,B)-r(t,z)\1\right) \right \Vert^2 = \sum_{i=1}^m \underbrace{\left \Vert \Sigma_i^{-1}\left(\eta_i + \la_i(B_i)\right) \right\Vert^2}_{=: \left(q_1(t,B)\right)_{i}} z_i\\
		&\Sigma^z(t,z)\rho(t,z)\Sigma^{-1}(t,z)\left(\mu(t,z)+\hat{\lambda}^{\ast}(t,z,B)-r(t,z)\1\right) \\
		&\qquad = \begin{pmatrix}
			\sigma_1 \sqrt{z_1} & & 0 \\
			& \ddots & \\
			0 & & \sigma_m \sqrt{z_m}
		\end{pmatrix}\begin{pmatrix}
			\rho_1' & & 0 \\
			& \ddots & \\
			0 & & \rho_m'
		\end{pmatrix}\begin{pmatrix}
			\left(\Sigma_1 \sqrt{z_1} \right)^{-1} & & 0 \\
			& \ddots & \\
			0 & & \left(\Sigma_m \sqrt{z_m} \right)^{-1}
		\end{pmatrix} \begin{pmatrix}
			\left(\eta_1 + \la_1(B_1)\right)z_1 \\
			\vdots \\
			\left(\eta_m + \la_m(B_m)\right)z_m 
		\end{pmatrix}\\
		&\qquad = \underbrace{\begin{pmatrix}
				\sigma_1 \rho_1'\Sigma_1^{-1}\left(\eta_1 +  \la_1(B_1)\right) &  0\\
				&\ddots \hfill \\
				0 &  \sigma_m \rho_m'\Sigma_m^{-1}\left(\eta_m +  \la_m(B_m)\right)
		\end{pmatrix}}_{=: g_1(t,B)}z.
	\end{align*}
	By setting the remaining coefficients $h_0,$ $l_0,$ $q_0,$ and $g_0$ as zero, Condition \ref{cond: Exp Separability optimal lambda*} is satisfied. Moreover, the ODEs (\ref{eq: constr ODE for A}) and (\ref{eq: constr ODE for B}) simplify to
	\begin{align*}
		A_{\tau}\left(\tau\right) &= br + \left(\kappa \odot \theta\right)'B(\tau) \\
		\left(B_{\tau}\left(\tau\right)\right)_i&= b\left(p_1(T-\tau,B(\tau))\right)_i + \left(k_1(T-\tau)B(\tau)\right)_i + \frac{1}{2}\left(B(\tau)'h_1[\cdot]B(\tau)\right)_i \\
		& \quad + \frac{1}{2}\frac{b}{1-b}\big[q_1(T-\tau, B(\tau)) + 2g_1(T-\tau, B(\tau))B(\tau) + B(\tau)'\left(l_1[\cdot]+ h_1[\cdot]\right)B(\tau)\big]_i \\
		&= b\delta_{K_i}\left(\la_i(B_i(\tau))\right)-\kappa_iB_i(\tau) + \frac{1}{2}\sigma_i^2\left(B_i(\tau)\right)^2 \\
		& \quad + \frac{1}{2}\frac{b}{1-b}\left[ \left \Vert \Sigma_i^{-1}\left(\eta_i + \la_i(B_i(\tau))\right) \right \Vert^2 + 2\sigma_iB_i(\tau)\rho_i'\Sigma_i^{-1}\left(\eta_i + \la_i(B_i(\tau))\right) + \sigma_i^2\Vert \rho_i \Vert^2 B_i(\tau)^2 \right] \\
		&= -\kappa_iB_i(\tau) + \frac{1}{2}\sigma_i^2\left(B_i(\tau)\right)^2 \\
		&\quad + \frac{1}{2}\frac{b}{1-b}\left[2(1-b)\delta_{K_i}\left(\la_i(B_i(\tau))\right) + \left \Vert \Sigma_i^{-1}\left(\eta_i + \la_i(B_i(\tau))\right) + \sigma_iB_i(\tau)\rho_i \right \Vert^2 \right] \\
		&= -\kappa_iB_i(\tau) + \frac{1}{2}\sigma_i^2\left(B_i(\tau)\right)^2 \\
		&\quad + \frac{1}{2}\frac{b}{1-b}\inf_{\lambda_i \in \R^{d_i}}\left \{ 2(1-b)\delta_{K_i}(\lambda_i) + \left \Vert \Sigma_i^{-1}\left(\eta_i + \lambda_i\right) + \sigma_iB_i(\tau)\rho_i \right \Vert^2  \right \}.
	\end{align*}
	Hence, according to Theorem \ref{thm: solution to dual HJB PDE given condition EAS},
	$$
	G(t,v,z) = \frac{1}{b}v^b\exp(A(T-t) + B(T-t)'z)
	$$
	is a solution to the dual HJBI PDE (\ref{eq: dual HJB PDE general}) and the corresponding candidate optimal portfolio (\ref{eq: optimal constr. pi (general)}) is given as in the statement of the lemma.\\
	Furthermore,
		\begin{align*}
			\big(&\Sigma(t,z)\Sigma(t,z)' \big)^{-1} \left(\Sigma^z(t,z)\rho(t,z)\Sigma(t,z)'\right)'B(T-t) \\
			&= \left(\Sigma(t,z)'\right)^{-1}\rho(t,z)'\Sigma^z(t,z)'B(T-t) \\
			&= \begin{pmatrix}
				\left(\Sigma_1'\sqrt{z_1}\right)^{-1} & & 0 \\
				& \ddots & \\
				0 & & 	\left(\Sigma_m'\sqrt{z_m}\right)^{-1}
			\end{pmatrix}\begin{pmatrix}
				\rho_1 & & 0 \\
				& \ddots & \\
				0 & & \rho_m
			\end{pmatrix}\begin{pmatrix}
				\sigma_1 \sqrt{z_1}B_1(T-t) \\
				\vdots \\
				\sigma_m \sqrt{z_m}B_m(T-t)
			\end{pmatrix} \\
			&= \begin{pmatrix}
				\sigma_1 B_1(T-t)\left(\Sigma_1^{-1}\right)'\rho_1 \\
				\vdots \\
				\sigma_m B_m(T-t)\left(\Sigma_m^{-1}\right)'\rho_m
			\end{pmatrix}
		\end{align*}
		and therefore	
		\begin{align*}
			\ubar{\pi}^{\ast}(t,v,z)&= \frac{1}{1-b}\left(\Sigma(t,z)\Sigma(t,z)'\right)^{-1}\left[\mu(t,z) + \ubar{\lambda}^{\ast}(t,v,z) - r(t,z) \1\right] \\
			& \quad 
			+\frac{1}{1-b}\left(\Sigma(t,z)\Sigma(t,z)'\right)^{-1}\left(\Sigma^z(t,z)\rho(t,z)\Sigma(t,z)'\right)'B(T-t)\\
			&= \frac{1}{1-b}\begin{pmatrix}
				(\Sigma_1\Sigma_1'z_1)^{-1}& & 0 \\
				& \ddots & \\
				0 & & (\Sigma_m\Sigma_m'z_m)^{-1}
			\end{pmatrix}
			\begin{pmatrix}
				\left(\eta_1  + \la_1(B_1(T-t))\right)z_1\\
				\vdots \\
				\left(\eta_m  + \la_m(B_m(T-t))\right)z_m
			\end{pmatrix} \\
			& \qquad + \frac{1}{1-b}\begin{pmatrix}
				\sigma_1 B_1(T-t)\left(\Sigma_1^{-1}\right)'\rho_1 \\
				\vdots \\
				\sigma_m B_m(T-t)\left(\Sigma_m^{-1}\right)'\rho_m
			\end{pmatrix} \\
			&= \frac{1}{1-b}
			\begin{pmatrix}
				(\Sigma_1\Sigma_1')^{-1}\Big(\eta_1  + \la_1(B_1(T-t)) + \sigma_1 B_1(T-t)\Sigma_1\rho_1 \Big) \\
				\vdots \\
				(\Sigma_m\Sigma_m')^{-1}\Big(\eta_m  + \la_m(B_m(T-t))+ \sigma_m B_m(T-t)\Sigma_m\rho_m \Big)
			\end{pmatrix}
	\end{align*}.
\end{proof}
\vspace{0.5cm}
\begin{proof}[Proof of Lemma \ref{lem: solution to dual HJBI PDE OU model}]
	In $\mathcal{M}_{OU},$ the minimizer $\hat{\lambda}^{\ast}$ from (\ref{eq: minimization w.r.t. lambda}) is given as
	\begin{align*}
		\hat{\lambda}^{\ast}(t,z,B) &= 	 \underset{\lambda \in \R^d}{\text{argmin}}\Bigg \{ 2(1-b)\delta_K(\lambda) + \Big \Vert \underbrace{\Sigma(t,z)^{-1}\big(\mu(t,z)-r(t,z)\1 + \lambda \big)}_{\eta + \left(b(t;\hat{T})'\sigma\right)^{-1}\lambda} + \underbrace{\left(\Sigma^z(t,z)\rho(t,z)\right)'}_{= \sigma'}B \Big\Vert^2 \Bigg \} \\ 
		&=  \underset{\lambda \in \R^d}{\text{argmin}}\left \{ 2(1-b)\delta_K(\lambda) +  \left \Vert \eta + \sigma'B + \left(b(t;\hat{T})'\sigma\right)^{-1}\lambda \right \Vert^2 \right \} \\*
		&= \la(t,B),
	\end{align*}
	i.e. $\hat{\lambda}^{\ast}(t,v,z) = \la(t,B)$ is a deterministic function independent of $z$, but dependent on $B.$ Considering the market coefficients in Condition \ref{cond: Exp Separability optimal lambda*}, we get
	\begin{align*}
		&\mu^{z}(t,z) = \kappa\odot \left(\theta-z\right) = \underbrace{\kappa\odot \theta}_{=:k_0(t)}+\underbrace{\begin{pmatrix}
				-\kappa_1 & & 0 \\
				& \ddots & \\
				0 & & - \kappa_m 
		\end{pmatrix}}_{=: k_1(t)}z \\
		&\Sigma^z(t,z)\Sigma^z(t,z)= \underbrace{\sigma \sigma'}_{=:h_0(t)} \\
		&\Sigma^z(t,z)\underbrace{\rho(t,z)}_{= I_m}\left(\Sigma^z(t,z) \underbrace{\rho(t,z)}_{=I_m}\right)'-	\Sigma^z(t,z)\Sigma^z(t,z)'= 0\\
		&r(t,z) + \delta_K(\hat{\lambda}^{\ast}(t,z,B)) = \underbrace{w_0 + \delta_K(\la(t,B))}_{=:p_0(t,B)} + \underbrace{w_1'}_{=: p_1(t,B)'}z \\
		&\left\Vert\Sigma^{-1}(t,z)\left(\mu(t,z)+\hat{\lambda}^{\ast}(t,z,B)-r(t,z)\1\right) \right \Vert^2 = \underbrace{\Vert\eta + \left(b(t; \hat{T})'\sigma\right)^{-1}\la(t,B)\Vert^2}_{=: q_0(t,B)}\\
		&\underbrace{\Sigma^z(t,z)}_{=\sigma}\underbrace{\rho(t,z)}_{=I_m}\underbrace{\Sigma^{-1}(t,z)\left(\mu(t,z)+\hat{\lambda}^{\ast}(t,z,B)-r(t,z)\1\right)}_{=\eta + \left(b(t;\hat{T})'\sigma\right)^{-1}\la(t,B)} = \underbrace{\sigma\eta + \left(b(t;\hat{T})'\right)^{-1}\la(t,B)}_{=:g_0(t,B)}.\\
	\end{align*}
	By setting the remaining coefficients $h_1,$ $l_0,$ $l_1,$ $q_1$ and $g_1$ as zero, Condition \ref{cond: Exp Separability optimal lambda*} is satisfied. Moreover, the ODEs (\ref{eq: constr ODE for A}) and (\ref{eq: constr ODE for B}) simplify to
	\begin{align*}
		A_{\tau}(\tau)&= b\left(w_0 + \delta_K(\la(T-\tau, B(\tau)))\right) + (\kappa \odot \theta)'B(\tau) + \frac{1}{2}\Vert \sigma'B(\tau)\Vert^2 \\
		& \qquad + \frac{1}{2}\frac{b}{1-b}\Bigg( \left \Vert\eta + \big(b(T-\tau; \hat{T})'\sigma\big)^{-1}\la(T-\tau,B(\tau))\right \Vert^2 \\
		&\hspace{2.5cm} +2 \left(\sigma\eta + \left(b(t;\hat{T})'\right)^{-1}\la(t,B(\tau))\right)'B(\tau) + \left \Vert \sigma'B(\tau) \right \Vert^2 \Bigg) \\
		&= bw_0 + (\kappa \odot \theta)'B(\tau) + \frac{1}{2}\Vert \sigma'B(\tau)\Vert^2 \\
		& \qquad + \frac{1}{2}\frac{b}{1-b}\Bigg(2(1-b) \delta_K(\la(T-\tau, B(\tau))) + \left \Vert \eta + \sigma'B(\tau) + \left(b(T-\tau;\hat{T})'\right)^{-1}\la(T-\tau, B(\tau))  \right \Vert^2 \Bigg) \\
		&= bw_0 + (\kappa \odot \theta)'B(\tau) + \frac{1}{2}\Vert \sigma'B(\tau)\Vert^2  +  \frac{1}{2}\frac{b}{1-b} \inf_{\lambda \in \R^d}\left(2(1-b) \delta_K(\lambda) + \left \Vert \eta + \sigma'B(\tau) + \left(b(T-\tau;\hat{T})'\right)^{-1}\lambda  \right \Vert^2 \right)
	\end{align*}
	and
	\begin{align*}
		B_{\tau}(\tau) &= bw_1 - \kappa\odot B(\tau).
	\end{align*}
	Hence, according to Theorem \ref{thm: solution to dual HJB PDE given condition EAS},
	$$
	G(t,v,z) = \frac{1}{b}v^b\exp(A(T-t) + B(T-t)'z)
	$$
	is a solution to the dual HJBI PDE (\ref{eq: dual HJB PDE general}). The corresponding candidate optimal portfolio (\ref{eq: optimal constr. pi (general)}) is given through
	\begin{align*}
		\ubar{\pi}^{\ast}(t,v,z)&= \frac{1}{1-b}\left(\Sigma(t,z)\Sigma(t,z)'\right)^{-1}\Big[\mu(t,z) + \ubar{\lambda}^{\ast}(t,v,z) - r(t,z) \1 + \left(\Sigma^z(t,z)\rho(t,z)\Sigma(t,z)'\right)'B(T-t) \Big]\\
		&= \frac{1}{1-b}\left(\sigma' b(t;\hat{T})\right)^{-1}\Big(\eta + \left( b(t;\hat{T})'\sigma\right)^{-1} \la(t,B(T-t)) + \sigma'B(T-t)\Big).
	\end{align*}
\end{proof}
\vspace{0.5cm}
\begin{proof}[Proof of Corollary \ref{cor: uniform integrability OU model}]
We again verify Condition \ref{cond: uniform integrability condition for dual} by showing the $L^q$ boundedness of $G\left(\tau^0_{n,t}, V^{v_0, \ubar{\pi}^{\ast}}_{0}(\tau^0_{n,t}),z(\tau^0_{n,t})\right)$ in $n\in \N$ for arbitrary $q>1.$ \\
As per Remark \ref{rem: explicit solution B OU Model}, there exists a closed-form expression for $B$ which is continuously differentiable. Moreover, the matrix $b(t;\hat{T})$ is continuously differentiable in $t$ and non-singular for all $t \in [0,T]$ and therefore uniformly bounded in $t\in [0,T].$ Following the same arguments as in the proof of Corollary \ref{cor: uniform integrability B.S. model}, this has the consequence that the minimizer $\la(t,B(T-t)),$ $A(T-t)$ and the candidate optimal portfolio $\pi^{\ast}(t):= \ubar{\pi}^{\ast}(t,v,z)$ are uniformly bounded in $t\in [0,T].$ 

For arbitrary $q>1,$ we can thus find a constant $C_q > 0$ such that for all $t\in [0,T]$
\begin{align}
	\Big | G(t,V^{v_0, \ubar{\pi}^{\ast}}(t),z(t)) \Big |^q &= \frac{1}{|b|}\exp\Bigg(bq\int_0^t w_0+w_1'z(s) + \eta'\sigma'b(s;\hat{T})\pi^{\ast}(s) - \frac{1}{2}\Vert \sigma'b(s;\hat{T})\pi^{\ast}(s) \Vert^2 ds \nonumber \\
	& \hspace{2cm}+ bq \int_0^t \pi^{\ast}(s)'b(s;\hat{T})'\sigma dW(s) + qA(T-t) + qB(T-t)'z(t)\Bigg) \nonumber \\
	&\leq C_q \exp \left(\underbrace{ bq \int_0^tw_1'z(s)ds + bq\int_0^t \pi^{\ast}(s)'b(s;\hat{T})'\sigma dW(s) + qB(T-t)'z(t)}_{=: X_t}\right) \nonumber \\*
	&= C_q \exp \left(X_t\right) \label{eq: first Lq bound value function OU}.
\end{align}
Since, $B$ is continuously differentiable, we can use It\^{o}'s product rule to rewrite 
\begin{align}
	B(T-t)'z(t) &= B(T)'z_0 + \int_0^t B(T-s)'dz(s) + \int_0^t z(s)'d\big(B(T-s)\big) + \underbrace{ \langle z, B(T-\cdot)\rangle_t}_{= 0} \nonumber\\
	&= B(T)'z_0 + \int_0^t B(T-s)'\kappa\odot \left[\theta - z(s)\right]-z(s)'B_{\tau}(T-s)ds  + \int_0^t B(T-s)'\sigma dW^z(s) \label{eq: expansion of B(T-t)z(t)}
\end{align}
Due to $\rho(t,z) = I_m,$ we know that $W^z(t) = W(t)$ holds $\mathcal{L}[0,T]\otimes Q$-a.e.. Hence, using (\ref{eq: expansion of B(T-t)z(t)}) and disregarding terms of finite variation, the quadratic variation of $X$ can be computed as
\begin{align*}
	\langle X \rangle_t &= \Bigg \langle  bq\int_0^tw_1'z(s)ds + bq\int_0^{\cdot} \pi^{\ast}(s)'b(s;\hat{T})'\sigma \underbrace{dW(s)}_{=dW^z(s)} + qB(T-\cdot)'z(\cdot) \Bigg \rangle_t \\
	&=  \left \langle  bq\int_0^{\cdot} \pi^{\ast}(s)'b(s;\hat{T})'\sigma +  \frac{1}{b}B(T-s)'\sigma dW^z(s)    \right \rangle_t \\
	&= b^2q^2 \int_0^t \left \Vert \sigma'\left(b(s;\hat{T})\pi^{\ast}(s) +\frac{1}{b} B(T-s)\right) \right \Vert^2  ds.
\end{align*}
Since all involved funtions are bounded and deterministic,  $\langle X \rangle_t \leq \langle X \rangle_T < \infty$ yields a deterministic upper bound on $\langle X \rangle_t$ for all $t\in [0,T].$ Therefore, we can continue equation (\ref{eq: first Lq bound value function OU}) to obtain for all $t\in [0,T]$
\begin{align}
	\Big | G(t,V^{v_0, \ubar{\pi}^{\ast}}(t),z(t)) \Big |^q &\overset{(\ref{eq: first Lq bound value function OU})}{\leq} C_q \exp(X_t) \leq \underbrace{C_q\exp\left(\frac{1}{2}\langle X \rangle_T\right)}_{=:\tilde{C}_q}\underbrace{\exp\left(X_t-\frac{1}{2}\langle X \rangle_t\right)}_{=:M_t}  = \tilde{C}_qM_t \label{eq: second Lq bound value function OU}
\end{align}
The process $M = \left(M_t\right)_{t\in [0,T]}$ is a non-negative local martingale and thus a supermartingale. Doob's optional sampling theorem (\enquote{O.S.}) implies
\begin{align*}
	\sup_{n\in \N} \E \left[ \left |G\left(\tau^0_{n,t}, V^{v_0, \ubar{\pi}^{\ast}}_{0}(\tau^0_{n,t}),z(\tau^0_{n,t})\right) \right |^q\right] \overset{(\ref{eq: second Lq bound value function OU})}{\leq}\tilde{C}_q \sup_{n\in \N} \E \left[ M_{\tau^0_{n,t}} \right] \overset{O.S.}{\leq} \tilde{C}_q M_0 = \tilde{C}_q < \infty.
\end{align*}
Hence, 
$$
\left(G\left(\tau^0_{n,t}, V^{v_0, \ubar{\pi}^{\ast}}_{0}(\tau^0_{n,t}),z(\tau^0_{n,t})\right)\right)_{n\in \N}
$$
is bounded in $L^q$ for any $q>1$ and $t\in [0,T]$ and is thus uniformly integrable for any $t\in [0,T]$ (see Theorem 4.5.9 in \cite{Bogachev2007} with $G(t) = t^q$).  Hence, Condition \ref{cond: uniform integrability condition for dual} is satisfied and $\ubar{\pi}^{\ast}$ is optimal for $\mathbf{(P)}$ by virtue of Theorem \ref{thm: verification theorem primal OP}.
\end{proof}
\vspace{0.5cm}

\newpage

\theendnotes

\end{document}